\documentclass[12pt]{article} %
\usepackage{amsmath,amssymb,amsthm}
\numberwithin{equation}{section} %
\theoremstyle{plain}
   \newtheorem{thm}{\hspace{\parindent}{\sc Theorem}}[section] %
   \newtheorem{pro}[thm]{\hspace{\parindent}Proposition}
   
   \newtheorem{lem}[thm]{\hspace{\parindent}Lemma}
\theoremstyle{remark} %
   \newtheorem{rem}{\hspace{\parindent}Remark}[section] %
\newcommand{\bA}{\mathbf{A}}

\newcommand{\bq}{\mathbf{q}}
\newcommand{\bR}{\mathbb{R}}
\newcommand{\bx}{\mathbf{x}}

\newcommand{\calF}{\mathcal{F}}
\newcommand{\Cts}{\mathcal{C}}

\newcommand{\Cspace}{C^{\infty}}
\newcommand{\Czerospace}{C^{\infty}_0(\bR^d)}
\newcommand{\domain}{[0,T]\times \bR^d}

\newcommand{\kdelta}{K_{\Delta}}
\newcommand{\ksdelta}{K_{s\Delta}}
\newcommand{\limepsilon}{\lim_{\epsilon\rightarrow 0+}}
\newcommand{\qdelta}{q_{\Delta}}
\newcommand{\qts}{q^{t,s}_{x,y}}

\newcommand{\Sspace}{{\cal S}}
\newcommand{\ts}{t,s}

\def\dbar{{\mathchar'26\mkern-12mud}}
\begin{document}
\title{On the mathematical formulation of the restricted Feynman path integrals through broken line paths}
\author{Wataru Ichinose\thanks{This work was supported by JSPS KAKENHI Grant Number JP18K03361.\endgraf 2020 Mathematics Subject Classification. Primary 46T12: Secondary 81P15, 81Q05. }} %
\date{}
\maketitle %
%
%
%
\begin{abstract}
The restricted Feynman path integrals (RFPIs) have been proposed to study continuous quantum measurements in physics. The RFPIs are heuristically determined in terms of the usual probability amplitude multiplied by weight for each path, which contains information about the results and the resolution of the measuring device. In the present paper we will consider the RFPIs particularly for the position measurements and will prove rigorously that these RFPIs are well defined in the $L^{2}$ space and are the solutions to the non-self-adjoint Schr\"odinger equations. Our results in the present paper give a generalization of the results on the usual Feynman path integrals for the Schr\"odinger equations.
 Furthermore, our results are extended to quantum spin systems.
\end{abstract}
%
%
%
%

\section{Introduction}%
Let $T > 0$ be an arbitrary constant, $0 \leq t \leq T$ and 
$x = (x_1,\dotsc,x_d)\in \bR^d$. First, we consider a one-particle system with mass $m > 0$ and charge $q_c \in \bR$ moving in $\bR^d$ with  electric strength $E(t,x) = (E_1,
\dots, E_d) \in \bR^d$ and a  magnetic strength tensor $B(t,x) = (B_{jk}(t,x))_{1\leq j < k \leq d}
\in \bR^{d(d-1)/2}$. 
Let   $(V(t,x),A(t,x)) = (V,A_1,
\dots,A_d)  \in \bR^{d+1}$ be an electromagnetic potential, i.e.
\begin{align} \label{1.1}
     & E = -\frac{\partial A}{\partial t} - \frac{\partial V}{\partial 
x},\notag \\
          & B_{jk} =  \frac{\partial A_k}{\partial x_j}  -\frac{\partial 
A_j}{\partial x_k}
\quad (1 \leq j <  k \leq d),
\end{align}
where  $\partial V/\partial x = (\partial V/\partial x_1,\dots,\partial V/\partial x_d)$. 
Then the Lagrangian function and the classical action are given by
\begin{equation} \label{1.2}
      {\cal L}(t,x,\dot{x})=  \frac{m}{2}|\dot{x}|^2 + q_c\dot{x}\cdot A(t,x) - q_cV(t,x),
\   \dot{x}\in \bR^d
\end{equation}
and 
\begin{equation} \label{1.3}
   S(t,s;q) = \int_s^t {\cal L}(\theta,q(\theta),\dot{q}(\theta))d\theta,\quad \dot{q}(\theta) = \frac{dq(\theta)}{d\theta}
\end{equation}
for a path $q(\theta) \in \bR^d\ (s \leq \theta \leq t)$, respectively.
The corresponding Schr\"odinger equation is defined by
\begin{align} \label{1.4}
& i\hbar \frac{\partial u}{\partial t}(t)  = H(t)u(t)\notag\\ 
& := \left[ \frac{1}{2m}\sum_{j=1}^d
      \left(\frac{\hbar}{i}\frac{\partial}{\partial x_j} - q_cA_j(t,x)\right)^2 + q_cV(t,x)\right]u(t),
\end{align}
where $\hbar$ is the Planck constant.
Throughout this paper we always consider solutions to the Schr\"odinger equations in the sense of distribution. Let $L^2 = L^2(\bR^d)$ denote the space of all square integrable functions in
$\bR^d$ with the inner
product $(f,g) := \int f(x)g(x)^*dx$ and the norm $\Vert f\Vert$, where $g(x)^*$ denotes the complex conjugate of $g(x)$. 
\par
Consider a continuous quantum   measurement of the position of the particle in the time interval $[0,T]$. Let $\{a(t) \in \bR^d;0 \leq t \leq T\}$ be its result and $\Delta a > 0$ its resolution or error  of the measuring device. The measurement gives a change of the probability amplitude of the particle, called wave-function reduction (cf. \S 17.5 of \cite{Greiner 1993} and \S 1.4 of \cite{Sakurai}). Let $f \in L^2$ be a probability amplitude of the particle at an initial time $t=0$.  Then, if we follow Feynman's postulates I and II on p. 371 of \cite{Feynman 1948}, the probability amplitude in the continuous measurement is heuristically given by the `` sum " of $e^{i\hbar^{-1}S(t,0;q)}f(q(0))$ over a set $\Gamma(t,x;\Delta a)$ of all paths $q$ satisfying $q(t) = x$ and $|q(\theta) - a(\theta)| \leq \Delta a$ for all $\theta \in [0,t]$, i.e.
\begin{equation} \label{1.4.2}
     \int_{\Gamma(t,x;\Delta a)} e^{i\hbar^{-1}S(t,0;q)}f(q(0)){\cal D}q.
\end{equation}
An alternative Feynman path integral description has been proposed by Mensky in \S 4.2 of \cite{Mensky 1993} and \S 5.1.3 of \cite{Mensky 2000}, written formally as
\begin{equation} \label{1.4.3}
     \int_{\Gamma(t,x)} e^{i\hbar^{-1}S(t,0;q) - c\int_{0}^{t}|q(\theta)-a(\theta)|^{2}d\theta/(\Delta a)^{2}}f(q(0)){\cal D}q
\end{equation}
with a constant $c > 0$ , where $\Gamma(t,x) := \Gamma(t,x;\infty)$.
See also \S 10.5.4 of \cite{Albeverio et all}, \cite{Dorofeyev}, \S 3.2 of \cite{Feynman-Hibbs}, \S 5.1 of \cite{Mazzucchi} and \cite{Mensky 2003}.  Both of \eqref{1.4.2} and \eqref{1.4.3} are called the restricted Feynman path integrals (RFPIs).
\par
	Our purpose in the present paper is to give a rigorous meaning in the $L^{2}$ space to each of  \eqref{1.4.3} and a more general formula with a weight function $W(t,x) \in \bR$ replacing $c|x - a(t)|^{2}/(\Delta a)^{2}$.  Furthermore, we wil show that each of \eqref{1.4.3} and the  more general forumula stated above is the solution to the non-self-adjoint Schr\"odinger equation with $f$ at $t = 0$.
\par
	As far as the author knows, we have been able to give a rigorous meaning  to \eqref{1.4.3} only for $A = 0$ and $V = C|x|^{2}$ with a constant $C \in \bR$, where we can directly calculate \eqref{1.4.3} by using Gaussina integrals (cf. \S 4.4 and \S 5.4 of \cite{Mensky 1993}).
	\par
	We also note that there is another approach to continuous quantum  position measurements of the particle.  We begin by  considering a sequence of $n$ instantaneous position measurements separated by a time $\Delta t$ and then, determine the evolution of the measured system in the continuous limits $n \to \infty$ and $\Delta t \to 0$ (cf. \cite{Caves 1986, Caves 1987-1,Caves 1987-2}, Chapter 3 in \cite{Jacobs}, Chapter 2 in \cite{Mensky 1993} and Chapter 2 in \cite{Mensky 2000}). 
\par
	Let $W(t,x)$ be a weight function and define the effective Lagrangian function under the measurement by 
\begin{equation} \label{1.5}
      {\cal L}_w(t,x,\dot{x}) =  {\cal L}(t,x,\dot{x}) + i\hbar W(t,x)
\end{equation}
and  the effective classical action by
\begin{align} \label{1.6}
   S_w(t,s;q) & =   \int_s^t {\cal L}_w(\theta,q(\theta),\dot{q}(\theta))d\theta \notag \\
   & = S(t,s;q) + i\hbar\int_{s}^{t}W(\theta,q(\theta))d\theta
\end{align}
for a path $q(\theta) \in \bR^d$ as in \S 4.2 of \cite{Mensky 1993}.  In the present paper we will prove for $W = c|x-a(t)|^{2}/(\Delta a)^{2}$ and more general weight functions that the RFPIs 
\begin{equation} \label{1.7}
 K(t,0)f := \int_{\Gamma(t,x)} e^{i\hbar^{-1}S_w(t,0;q)}f(q(0)){\cal D}q
\end{equation}
are well defined in $L^{2}$ for $f \in L^{2}$. Furthermore, we will prove that $K(t,0)f$ satisfy
the non-self-adjoint Schr\"odinger equations, derived  from \eqref{1.5} through the Legendre transformation,
\begin{align} \label{1.8}
& i\hbar \frac{\partial u}{\partial t}(t)  = H_w(t)u(t)\notag \\
& := \left[ \frac{1}{2m}\sum_{j=1}^d
      \left(\frac{\hbar}{i}\frac{\partial}{\partial x_j} - q_cA_j(t,x)\right)^2 + q_cV(t,x) -i\hbar W(t,x)\right]u(t)
\end{align}
with $u(0) = f$,
which was suggested for \eqref{1.4.3} in \S 4.3.1 of \cite{Mensky 2000}.
\par
	Next, we will generalize the above results  to a one-particle spin system, where all spin components or directions may move separately in $\bR^d$ as in the Stern-Gerlach experiment (cf. Chap. 12 of \cite{Greiner 1993} and \S 1.1 of \cite{Sakurai}).  We generally suppose that  a particle has $l$ spin components (cf. p.12 in \cite{Bethe et all} and  \S 2.2 of \cite{Greiner 1989}) and we consider a  continuous position measurement for all spin components in  $[0,T]$, where $l \geq 0$ is an integer.  Although we don't know the physical meaning precisely, we will study
the effective Lagrangian function given by
\begin{equation} \label{1.9}
\mathcal{L}_{sw}(t,x,\dot{x}) = \mathcal{L}_{w}(t,x,\dot{x}) - \hbar H_s(t,x) + i\hbar W_s(t,x)
 \end{equation}
as a  generalization of \eqref{1.5},
where $\mathcal{L}_{w}(t,x,\dot{x})$ is the Lagrangian function defined by \eqref{1.5}, $H_s(t,x)$ an $l\times l$ Hermitian matrix denoting the spin term and  $W_s(t,x)$ an $l\times l$ Hermitian matrix denoting the weight term acting on the spin components.
The corresponding non-self-adjoint Schr\"odinger equation is written as
\begin{equation}  \label{1.10}
   i\hbar\frac{\partial u}{\partial t}(t) = \bigl[H_w(t)I + \hbar H_{s}(t,x)  - i\hbar W_s(t,x)\bigr]u(t), 
\end{equation}
where $H_w(t)$ is the Hamiltonian operator defined by \eqref{1.8}.  We will prove that the RFPI for \eqref{1.9} can be defined  rigorously in $(L^2(\bR^{d}))^l$ and is the solution to the equation \eqref{1.10}.  It is noted that for $W_s(t,x)$ we  assume 
\begin{equation}  \label{1.11}
  |\partial_{x}^{\alpha} w_{sij}(t,x)| \leq C_{\alpha}, \quad i, j = 1,2,\dots,l
\end{equation}
in $\domain$ for all $\alpha$, where $w_{sij}(t,x)$ denotes the $(i,j)$-component of $W_s(t,x)$.  
\par
%
	Finally, we  consider a quantum spin system consisting of $N$ particles with $l$ spin components each, under a  continuous position measurement for all spin components of all particles in $[0,T]$.
		\par
	We note that if $W(t,x) = 0$ and $W_s(t,x) = 0$ in \eqref{1.5} and \eqref{1.9}, all  results in the present paper give the same results as for the usual Feynman path integrals in \cite{Ichinose 1999, Ichinose 2007, Ichinose 2020}.
\par
     In the present paper  the RFPIs are defined by the time-slicing method in terms of piecewise free moving paths or broken line paths.  This  approach  to the Feynman path integrals are  widely used in the physics literature (cf.  \S 2.4 in \cite{Feynman-Hibbs}, \S 3.2 in \cite{Mensky 1993}, Appendix A3 in \cite{Mensky 2000}, \S 9.1 in \cite{Peskin-Schroeder} and \S 5.1 in \cite{Ryder}).
    \par
 	 We will prove the main theorems in the present paper, following the proofs in \cite{Ichinose 1997, Ichinose 1999, Ichinose 2007, Ichinose 2020}, where the usual Feynman path integrals, i.e. with $W(t,x) = 0$ and $W_s(t,x) = 0$ were studied.  More specifically, we first introduce the fundamental operator $\mathcal{C}(t,s)$ in \S 3, and then prove its stability and consistency.  Combining these results and the results in \cite{Ichinose 2019} concerning the non-self-adjoint Schr\"odinger equations \eqref{1.8} and \eqref{1.10}, we can complete the proofs of  our results.  In particular, we define the RFPIs for the spin system, following \cite{Ichinose 2007}.  We also note that in the  present paper we will use 
  the following delicate result  concerning the $L^2$-boundedness of pseudo-differential operators, which is stated as Theorem 13.13 on p. 322 in \cite{Zworski}. 
  \par
 \vspace{0.5cm}
 T{\sc heorem} 1.A.  {\it Suppose $p(x,\xi,x') \in S^0(\bR^{3d})$, i.e. 
 \begin{equation} \label{1.12}
\sup_{x,\xi,x'}|\partial_{\xi}^{\alpha}\partial_{x}^{\beta}\partial_{x'}^{\gamma}p(x,\xi,x')| \leq C_{\alpha\beta\gamma} < \infty
 \end{equation}
for all multi-indices $\alpha, \beta$ and $\gamma$, where   $\partial_{\xi}^{\alpha}$ denotes   
  $(\partial/\partial \xi_1)^{\alpha_1}
\cdots (\partial/\partial \xi_d)^{\alpha_d}$ for  $\alpha = (\alpha_1,\dots,\alpha_d)$.
  Let $P(X,\hbar D_x,X')$ be the pseudo-differential operator defined by 
 \begin{equation} \label{1.12.2}
\int e^{ix\cdot \xi}\ \dbar\xi \int e^{-ix'\cdot \xi}p(x,\hbar\xi,x')f(x')dx',\quad \dbar\xi = (2\pi)^{-d}d\xi
 \end{equation}
for $f \in \Sspace(\bR^d)$, where $x\cdot \xi = \sum_{j=1}^ d x_j\xi_j $ and $\Sspace(\bR^d)$ denotes the Schwartz space of all rapidly decreasing functions in $\bR^d$.  Then we have
 \begin{equation} \label{1.13}
\Vert P(X,\hbar D_x,X')\Vert_{L^2\to L^2} = \sup_{x,\xi,x'}|p(x,\xi,x')| + O(\hbar),
 \end{equation}
where $\Vert P\Vert_{L^2\to L^2}$ denotes the operator norm from $L^2$ into $L^2$.}
\vspace{0.5cm}
\par
The plan of the present paper is as follows.  In \S 2 all main results are  stated in Theorems 2.1 - 2.6.  In \S 3 and \S 4 we will prove the stability and the consistence of $\Cts(\ts)$ respectively.  In \S 5 Theorems 2.1 and 2.2 will be proved. In \S 6 Theorems 2.3 - 2.6 will be proved.
 \section{Main theorems}
Hereafter we suppose $\hbar = 1$ and $q_c = 1$ for  simplicity. 
We first consider \eqref{1.7}.
Let $t$ in $[0,T]$.  For an arbitrary integer $\nu \geq 1$ we take $\tau_j \in [0,t]\ (j = 1,2,\dots,\nu-1)$ satisfying $0 = \tau_0 < \tau_1 < \dots < \tau_{\nu-1} < \tau_{\nu}= t$,  set $\Delta := \{\tau_j\}_{j=1}^{\nu-1}$ and write $|\Delta|:= \max\{\tau_{j+1}- \tau_j; j = 0,1,\dots,\nu-1\}$. Let $x \in \bR^d$ be fixed.
  We take arbitrary points
$x^{(j)} \in \bR^d\ (j = 0,1,\dotsc,\nu-1)$ and determine the piecewise free moving path or the piecewise straight line $\qdelta(\theta;x^{(0)},\dotsc,x^{(\nu-1)},x) \in \bR^d \ (0 \leq \theta \leq t)$ by joining  $x^{(j)}$ at $\tau_j\ (j = 0,1, \dotsc,\nu, 
x^{(\nu)} = x)$  in order.  Let   $S_w(t, s; q)$ be the effective classical action defined by \eqref{1.6}.  Take $\chi \in \Cspace_0(\bR^d)$, i.e. an infinitely differentiable function on $\bR^d$ with compact support,  such that $\chi(0) = 1$ and  fix it through the present paper. For simplicity we suppose that $\chi$ is real-valued.  We will determine the approximation  of the RFPI expressed as \eqref{1.7}  by 
 \begin{align} \label{2.1}
\kdelta(t,0)f =  \limepsilon & \prod_{j=0}^{\nu-1}\sqrt{\frac{m}{2\pi i(\tau_{j+1} - 
\tau_{j})}}^{\ d}
        \int\cdots\int_{\bR^d} e^{iS_w(t,0;\qdelta)} \notag \\
       & \times f(x^{(0)}) \prod_{j=1}^{\nu-1}\chi(\epsilon x^{(j)})
          dx^{(0)}dx^{(1)}\cdots dx^{(\nu-1)}
 \end{align}
for $f \in \Cspace_0(\bR^d)$. 
The right-hand side of \eqref{2.1} is called an oscillatory integral and will be denoted by
 \begin{align*} 
  & \prod_{j=0}^{\nu-1}\sqrt{\frac{m}{2\pi i(\tau_{j+1} - 
\tau_{j})}}^{\ d}
      \text{Os} -  \int\cdots\int_{\bR^d} e^{iS_w(t,0;\qdelta)}  f(x^{(0)}) 
          dx^{(0)}dx^{(1)}\cdots dx^{(\nu-1)}
 \end{align*}
 (cf. p. 45 of \cite{Kumano-go}).
  \par
   For a multi-index $\alpha = (\alpha_1,\dots,\alpha_d)$ and $x \in \bR^d$ we  write $|\alpha| = 
\sum_{j=1}^d
\alpha_j$, $x^{\alpha} =  x_1^{\alpha_1}
\cdots  x_d^{\alpha_d}$ and  $<x> = \sqrt{1+|x|^2}$.
 In the present paper we often use symbols $C, C_{\alpha}, C_{\alpha\beta}$, $C_a$  and $\delta_{\alpha}$ to write down constants, though these values are different in general. 
 \par
   \vspace{0.5cm}
Throughout the present paper we assume that $\partial_x^{\alpha}E_j(t,x)\ (j = 1,2,\dots,d), \\ \partial_x^{\alpha}B_{jk}(t,x)\ (1 \leq j < k \leq d)$ and $\partial_x^{\alpha}W(t,x)$ are continuous in $\domain$ for all $\alpha$.  Then, $\partial_x^{\alpha}\partial_tB_{jk}(t,x)\ (1 \leq j < k \leq d)$ are also continuous in $\domain$ for all $\alpha$, because of Faraday's law $\partial_tB_{jk} =  - \partial E_k/\partial x_j + \partial E_j/\partial x_k$, which follows from \eqref{1.1}.
\par
{\bf Assumption 2.A.}  We assume 
\begin{equation} \label{2.2}
|\partial_x^{\alpha}E_j(t,x)| \leq C_{\alpha},\  |\alpha| \geq 1,\quad j = 1,2,\dots,d,
\end{equation}
\begin{equation} \label{2.3}
|\partial_x^{\alpha}B_{jk}(t,x)| \leq C_{\alpha}<x>^{-(1+ \delta_{\alpha})},\  |\alpha| \geq 1, \quad 1 \leq j < k \leq d
\end{equation}
in $\domain$ with constants $C_{\alpha} \geq 0$ and $\delta_{\alpha} > 0$.
\par
{\bf Assumption 2.B.}  We assume \eqref{2.2} and 
\begin{equation} \label{2.4}
|\partial_x^{\alpha}\partial_tB_{jk}(t,x)| \leq C_{\alpha}<x>^{-(1+ \delta_{\alpha})},\  |\alpha| \geq 1, \quad 1 \leq j < k \leq d
\end{equation}
in $\domain$ with constants $C_{\alpha} \geq 0$ and $\delta_{\alpha} > 0$.
\par
{\bf Assumption 2.C.} We assume that $\partial_x^{\alpha}A_j(t,x)\ (j = 1,2,\dots,d)$ and $\partial_x^{\alpha}V(t,x)$ are continuous in $\domain$ for all $\alpha$ and satisfy
\begin{equation} \label{2.5}
|\partial_x^{\alpha}A_j(t,x)| \leq C_{\alpha},\ |\alpha| \geq 1,\quad j = 1,2,\dots,d,
\end{equation}
\begin{equation} \label{2.6}
|\partial_x^{\alpha}V(t,x)| \leq C_{\alpha}<x>,\quad |\alpha| \geq 1
\end{equation}
in $\domain$ with constants $C_{\alpha} \geq 0$.
\par
{\bf Assumption 2.D.}  We assume 
\begin{equation} \label{2.7}
W(t,x) \geq - C(W),
\end{equation}
\begin{equation} \label{2.8}
|\partial_x^{\alpha}W(t,x)|^{p_{\alpha}} \leq C_{\alpha}\bigl\{1 + C(W) + W(t,x)\bigr\},\quad |\alpha| \geq 1,
\end{equation}
\begin{equation} \label{2.9}
|\partial_x^{\alpha}W(t,x)| \leq C_{\alpha}<x>,\quad |\alpha| \geq 1
\end{equation}
in $\domain$ with constants $C(W) \geq 0, C_{\alpha} \geq 0$ and $p_{\alpha} > 1$.
\par
{\it Example 2.1.} The function $c|x - a(t)|^2/(\Delta a)^2$  in \eqref{1.4.3} with a continuous path $a(t) \in \bR^{d}$ satisfies Assumption 2.D. 
\par
\begin{thm}
Suppose that Assumptions 2.A and 2.D are satisfied.  Then, there exists a constant $\rho^* > 0$ such that  the following statements hold for arbitrary potentials $(V,A)$ with continuous 
$V, \partial V/\partial x_j, \partial A_j/\partial t, \partial A_j/\partial x_k\ (j,k = 1,2,\dots,d)$ %
in $\domain$, all $\Delta$ satisfying $|\Delta| \leq \rho^*$ and all $t \in [0,T]$:
\\
(1) $K_{\Delta}(t,0)f$ defined on $f \in \Czerospace$ by \eqref{2.1} is
determined independently of the choice of $\chi$ and $K_{\Delta}(t,0)f$ can be uniquely extended to a bounded operator on $L^{2}$.
\\
(2) For all $f \in L^{2}$, as $|\Delta| \to 0$, $\kdelta(t,0)f$ converges in $L^{2}$ uniformly in $t \in [0,T]$ to an element $K(t,0)f \in L^{2}$, which we call the  RFPI of $f$. 
\\
(3) For all $f \in L^{2}$, $K(t,0)f$  belongs to   $C_{t}^{0}([0,T];L^{2})$, where $C^j_t([0,T];L^2)$ denotes the space of all $L^2$-valued, $j$-times continuously differentiable functions in $t \in [0,T]$.  In addition, $K(t,0)f$ is the unique solution  in $C^{0}_{t}([0,T];L^{2})$ to \eqref{1.8} with $u(0) = f$. 
\\
(4) Let $\psi(t,x) \in C^1(\domain)$ be a real-valued function such that $\partial_{x_j}\partial_{x_k}\psi(t,x)$ and $\partial_{t}\partial_{x_j}\psi(t,x)$ $(j,k = 1,2,\dotsc,d)$ are continuous in $[0,T] \times \bR^d$ and consider the gauge transformation 
\begin{equation}  \label{2.10}
    V' = V -\frac{\partial\psi}{\partial  t}, \quad  A'_j = A_j + \frac{\partial\psi}{\partial  x_j}\quad (j= 1,2,\dots,d).
\end{equation}
We write \eqref{2.1} for this $(V',A')$  as $K'_{\Delta}(t,0)f$.  Then we have the formula 
\begin{equation}  \label{2.11}
     K'_{\Delta}(t,0)f  = e^{i\psi (t,\cdot)}K_{\Delta}(t,0)\left(e^{-i\psi (0,\cdot)}f\right)
\end{equation}
for  all $f \in L^2$ as in the case of $W(t,x) = 0$ (cf. (6.16) in \cite{Ichinose 1999}), and we have the analogous relation between the limits $K'(t,0)f$ and $K(t,0)f$.
\end{thm}
Let us introduce the weighted Sobolev spaces
\begin{align} \label{2.12}
& B^a(\mathbb{R}^d)  := \{f \in  L^2(\mathbb{R}^d);
 \|f\|_a := \|f\| + \sum_{|\alpha| =  a}\bigl(\|x^{\alpha}f\| +\|\partial_x^{\alpha}f\|\bigr) < \infty\} \notag\\
 &  (a = 1,2,\dots)
 \end{align}
as in \cite{Ichinose 1999}. We denote the dual space of $B^a$ by $B^{-a}$ and the $L^2$ space by $B^0$. 
%
\begin{thm}
Suppose that either Assumption 2.A or 2.B is satisfied.  In addition, we suppose Assumptions 2.C and 2.D. Then there exists another constant $\rho^* > 0$  such that the same statements (1) - (4) as in Theorem 2.1 hold for all $\Delta$ satisfying $|\Delta| \leq \rho^*$ and all $t \in [0,T]$.  In addition, for all $f \in B^a(\bR^d)\ (a = 1,2,\dots)$ $K_{\Delta}(t,0)f$  belongs to   $B^a$ and  as $|\Delta| \to 0$, $\kdelta(t,0)f$ converges in $B^a$ uniformly in $t \in [0,T]$ to $K(t,0)f$, which belongs to  $C^0_t([0,T];B^a)$.  
\end{thm} 
	Next, we consider a one-particle spin system \eqref{1.9} with $l$ spin components.    Throughout the present paper we assume that $\partial_x^{\alpha}h_{sij}(t,x)$ and $\partial_x^{\alpha}w_{sij}(t,x)$ are continuous in $\domain$ for all $\alpha$ and $i,j = 1,2,\dots,l$, where $h_{sij}$  denotes the 
$(i,j)$-component of $H_s$. 
For a continuous path $q(\theta) \in \bR^d\ (s \leq \theta \leq t)$ we define an $l \times l$ matrix $\mathcal{F}(\theta,s;q)\ (s \leq \theta \leq t)$ by the solution to 
\begin{equation} \label{2.13}
\frac{d}{d\theta}\,\mathcal{U}(\theta) = -\bigl\{iH_s(\theta,q(\theta)) + W_s(\theta,q(\theta))\bigr\}\mathcal{U}(\theta), \quad \mathcal{U}(s) = I
 \end{equation}
with the identity matrix $I$.
\par
	Let $\Delta = \{\tau_j\}_{j=1}^{\nu - 1}$ be a subdivision of $[0,t]$ and $\qdelta = \qdelta(\theta;x^{(0)},\dotsc,x^{(\nu-1)},x) \in \bR^d \ (0 \leq \theta \leq t)$ the piecewise free moving path defined in the early part of this section.  We define the probability amplitude by 
\begin{equation} \label{2.14}
\exp *iS_{sw}(t,0;\qdelta) := \bigl(\exp iS_w(t,0;\qdelta)\bigr)\mathcal{F}(t,0;\qdelta)
 \end{equation}
 as in the case of $W_s(t,x) = 0$ (cf. \S 2 in \cite{Ichinose 2007}), where $S_w(t,0;\qdelta)$ is the classical action defined by \eqref{1.6}.  For $f = {}^t(f_1,\dots,f_l) \in \Czerospace^l$ we determine the approximation of the RFPI for this system under the measurement by
%
 \begin{align} \label{2.15}
\ksdelta(t,0)f =  \limepsilon & \prod_{j=0}^{\nu-1}\sqrt{\frac{m}{2\pi i(\tau_{j+1} - 
\tau_{j})}}^{\ d}
        \int\cdots\int_{\bR^d} \bigl(\exp *iS_{sw}(t,0;\qdelta)\bigr) \notag \\
       & \times f(x^{(0)}) \prod_{j=1}^{\nu-1}\chi(\epsilon x^{(j)})
          dx^{(0)}dx^{(1)}\cdots dx^{(\nu-1)}
 \end{align}
 as in \cite{Ichinose 2007}.
The $L^2$-norm  of $f = {}^t(f_1,\dots,f_l) \in (L^2)^l$ is defined by $\Vert f \Vert := \sqrt{\sum_{j=1}^l \Vert f_j\Vert^2}$.
\begin{thm}
Besides Assumptions 2.A and 2.D we suppose that $H_s(t,x)$ and $W_s(t,x)$ satisfy 
\begin{equation}  \label{2.16}
  |\partial_{x}^{\alpha} h_{sij}(t,x)| \leq C_{\alpha},\quad i, j = 1,2,\dots,l
\end{equation}
 and \eqref{1.11} for all $\alpha$, respectively.  Let $\rho^* > 0$ be the constant determined in Theorem 2.1.
Then the same statements for $\ksdelta(t,0)f$ as for $\kdelta(t,0)f$ in Theorem 2.1 hold, where $K_s(t,0)f := \lim_{|\Delta|\to 0}\ksdelta(t,0)f \in C^0_t
([0,T]
;(L^2)^l)$ for $f \in (L^2)^l$ is the unique solution in $C^0_t([0,T];(L^2)^l)$ to \eqref{1.10}
with $u(0) = f$.  
\end{thm}
{\it Example 2.2.} We consider a continuous quantum measurement of the positions of all  spin components of a particle.   Let $a^{(j)}(t)\in \bR^d\ (j = 1,2,\dots,l)$ be the result for the $j$-th spin component  and $\Delta a$ the resolution of the measuring device.
Then  in Theorem 2.3 we take $W(t,x) = 0$ and the diagonal matrix $W_s(t,x)$ with
\begin{equation*}
w_{sjj}(t,x) = \Omega \left(\frac{c|x - a^{(j)}(t)|^2}{(\Delta a)^2}\right), \quad j = 1,2,\dots,l,
\end{equation*}
where  $\Omega (\theta) \in \Cspace([0,\infty))$ is an increasing function such that $\Omega (\theta) = \theta$ if $0 \leq \theta \leq 1$ and  $\Omega (x) = L$ if $\theta \geq 2$ with a sufficiently large constant $L > 0$.  These $W(t,x)$ and $W_s(t,x)$ satisfy the assupmtions of Theorem 2.3.
\par
\vspace{0.3cm}
The $B^a$-norm  of $f = {}^t(f_1,\dots,f_l) \in (B^a)^l$ is defined by $\Vert f \Vert_a := \sqrt{\sum_{j=1}^l \Vert f_j\Vert_a^2}$. 
\begin{thm}
We suppose that either Assumption 2.A or 2.B is satisfied.  In addition, we suppose  Assumptions 2.C, 2.D, \eqref{1.11} and \eqref{2.16}.  Let $\rho^* > 0$ be the constant determined in Theorem 2.2.  Then the same statements  for $\ksdelta(t,0)f$ as for $\kdelta(t,0)f$ in Theorem 2.2 hold, where we replace $B^a$ with $(B^a)^l$.
\end{thm}
	Finally, we consider a quantum spin system consisting of $N$ particles which have $l$ spin components each.  We  perform a continuous quantum measurement of the positions of all spin components of all particles in $[0,T]$.  Denoting the coordinates of the $j$-th particle by $\bx_j \in \bR^d \ (j = 1 ,2,\dots,N)$
, we write $x = (\bx_1,\bx_2,\dots,\bx_N) \in \bR^{dN}$.
	 Let $W_j(t,\bx_j) \in \bR$ and set 
	\begin{align} \label{2.17}
& {\cal L}^{\sharp}_{w}(t,x,\dot{x}) =  \sum_{j=1}^N 
 \biggl\{\frac{m}{2}|\dot{\bx}_j|^2 + \dot{\bx}_j\cdot \bA_j(t,\bx_j) -V_j(t,\bx_j) \notag \\
&\qquad + iW_j(t,\bx_j) \biggr\}
   - \sum_{j,k=1,j \not= k}^N V_{jk}(t,\bx_j-\bx_k),
   \end{align}
where $\bA_j(t,\bx_j) \in \bR^d, V_j(t,\bx_j) \in \bR$ and $V_{jk}(t,\bx_j-\bx_k) \in \bR$.   The effective Lagrangian function we consider is
\begin{align} \label{2.18}
& {\cal L}^{\sharp}_{sw}(t,x,\dot{x}) =  {\cal L}^{\sharp}_{w}(t,x,\dot{x}) + \sum_{j=1}^N I_1\otimes \cdots \otimes I_{j-1} \notag \\
 & \qquad \otimes \bigl\{-H_{sj}(t,\bx_j) + iW_{sj}(t,\bx_j) \bigr\}\otimes I_{j+1}\otimes \cdots\otimes I_N,
   \end{align}
generalizing \eqref{1.9}, where $H_{sj}(t,\bx_j)$ and $W_{sj}(t,\bx_j)$ are $l \times l$ Hermitian matrices. 
 \par
 For a continuous path $\bq_j(\theta) \in \bR^d\ (j = 1,2,\dots,N, s \leq \theta \leq t)$, we define $\mathcal{F}_j(\theta,s;\bq_j)\ (s \leq \theta \leq t)$ by the solution to \eqref{2.13} where $H_s = H_{sj}$ and $W_s = W_{sj}$.
 For the piecewise free moving path $\qdelta = \bigl(\bq_{1\Delta}(\theta;\bx_1^{(0)},\dotsc,\bx_1^{(\nu-1)},\bx_1),\dots, \\
  \bq_{N\Delta}
  (\theta;\bx_N^{(0)},\dotsc,\bx_N^{(\nu-1)},\bx_N)\bigr)\in \bR^{Nd} \ (0 \leq \theta \leq t)$,  we define the probability amplitude by
\begin{align} \label{2.19}
& \exp *iS^{\sharp}_{sw}(t,0;\qdelta) \notag 
\\ &:= \bigl(\exp iS^{\sharp}_w(t,0;\qdelta)\bigr)\mathcal{F}_1(t,0;\bq_{1\Delta})\otimes \cdots \otimes \mathcal{F}_N(t,0;\bq_{N\Delta})
 \end{align}
in terms of the tensor product of matrices,
where $S^{\sharp}_w(t,0;\qdelta)$ is the classical action defined from \eqref{2.17}.  Then we 
determine the approximation $K^{\sharp}_{s\Delta}(t,0)f$ of the RFPI by \eqref{2.15}, where $\exp *iS_{sw}(t,0;\qdelta)$ and $f \in \Czerospace^l$ are replaced with  $\exp *iS^{\sharp}_{sw}(t,0;\qdelta)$ and $f = f_1 \otimes \cdots \otimes f_N$ $\bigl(f_j \in \Czerospace^l, j = 1,2,\dots,N\bigr)$, respectively.
\par
	Writing $A(t,x) = \bigl(\bA_1(t,\bx_1),\dots, \bA_N(t,\bx_N)\bigr) \in \bR^{dN}$ and $V(t,x) = \sum_{j=1}^N 
	V_j(t,\\
	\bx_j) 
		+ \sum_{j,k=1,j \not= k} V_{jk}(t,\bx_j-\bx_k) \in \bR$, we define $E(t,x) \in \bR^{dN}$ and $B_{jk}(t,x) \in \bR\  (1 \leq j < k \leq dN)$ by \eqref{1.1}.  Then we have the following.
	\begin{thm}
	Suppose that Assumptions 2.A is satisfied.  In addition, we assume that each $W_j(t,\bx_j)\ ( j = 1,2,\dots,N)$ satisfies 2.D and that each $W_{sj}(t,\bx_j)$  and $H_{sj}(t,\bx_j)$ satisfies \eqref{1.11} and \eqref{2.16}, respectively.  Let $(L^2)^l\otimes \cdots \otimes(L^2)^l$ denote the tensor product of $N$ copies of $L^2(\bR^d)^l$.
	Then, there exits a constant $\rho'^* > 0$ such that the same statements for $K^{\sharp}_{s\Delta}(t,0)f$ as for $\kdelta(t,0)f$ in Theorem 2.1 hold, where $K^{\sharp}_{s}(t,0)f := \lim_{|\Delta| \to 0} K^{\sharp}_{s\Delta}(t,0)f \in C^0_t([0,T];(L^2)^l\otimes \cdots \otimes(L^2)^l)$ for $f \in (L^2)^l\otimes \cdots \otimes(L^2)^l$ is the unique solution in $C^0_t([0,T];(L^2)^l\otimes \cdots \otimes(L^2)^l)$ to 
\begin{align} \label{2.20}
& i \frac{\partial u}{\partial t}(t)  =\Biggl[\sum_{j=1}^N \Bigg\{\frac{1}{2m}
      \left|\frac{1}{i}\frac{\partial}{\partial \bx_j} - \bA_j(t,\bx_j)\right|^2 + V_j(t,\bx_j) -iW_{j}(t,\bx_j)\notag\\ 
& \quad   + I_1\otimes \cdots \otimes I_{j-1}\otimes \bigl\{H_{sj}(t,\bx_j) -  iW_{sj}(t,\bx_j)\bigl\}\otimes I_{j+1}\otimes \cdots\otimes I_N\Biggr\}  \notag \\
&\qquad  + \sum_{j,k=1,j \not= k}^N V_{jk}(t,\bx_j-\bx_k)
\Biggr]u(t)
\end{align}
with $u(0) = f$.
		\end{thm}
		We see that the $N$-fold tensor product $L^2(\bR^d)^l\otimes \cdots \otimes L^2(\bR^d)^l$ is equal to $\left(L^2(\bR^d)\otimes \cdots \otimes L^2(\bR^d)\right)^{l^N}$ because we have
\begin{equation*}
  \begin{pmatrix}
  g_1(\bx_1) \\ g_2(\bx_1)
    \end{pmatrix}
    \otimes
    \begin{pmatrix}
  h_1(\bx_2) \\ h_2(\bx_2)
    \end{pmatrix}
    = \sum_{i,j = 1}^2g_i(\bx_1)h_j(\bx_2)e_i\otimes e_j
\end{equation*}
with $e_1 = {}^t(1,0)$ and $e_2 = {}^t(0,1)$, for examle, when $N = 2$ and $l = 2$. This shows 
\begin{equation*}
 L^2(\bR^d)^l\otimes \cdots \otimes L^2(\bR^d)^l 
= L^2(\bR^{dN})^{l^N}
\end{equation*}
because of $L^2(\bR^d)\otimes \cdots \otimes L^2(\bR^d) = L^2(\bR^{Nd})$
(cf. II.10 on p.52 in \cite{Reed-Simon I}).
		 In the same way we can define a subspace $B^a(\bR^{dN})^{l^N} (a = 1,2,\dots)$ in $L^2(\bR^d)^l\otimes \cdots \otimes L^2(\bR^d)^l$.  Then we have the following.
\begin{thm}
Suppose that either Assumption 2.A or 2.B is satisfied.  In addition, we suppose  Assumption 2.C for $(V(t,x),A(t,x))$  and that each $W_j(t,\bx_j), W_{sj}(t,\bx_j)$ and $H_{sj}(t,\bx_j)$ satisfies the assumptions stated in Theorem 2.5.  Then, there exits another constant $\rho\,'^* > 0$ such that  the same statements for $K^{\sharp}_{s\Delta}(t,0)f$ as for $\kdelta(t,0)f$ in Theorem 2.2 hold, where we replace $B^a(\bR^d)$ with  $B^a(\bR^{dN})^{l^N}$.  
\end{thm}
%
\begin{rem}
We consider polynomially growing potentials 
\begin{equation} \label{2.21}
V(t,x) = |x|^{2(l_0+1)} + \sum_{|\alpha| \leq 2l_0 + 1}a_{\alpha}(t)x^{\alpha},
\end{equation}
\begin{equation} \label{2.22}
A_j(t,x) =  \sum_{|\alpha| \leq l_0}b_{j\alpha}(t)x^{\alpha}\quad (j = 1,2,\dots,d)
\end{equation}
with an integer $l_0 \geq 1$ and functions $a_{\alpha}(t) \in \bR, b_{j\alpha}(t) \in \bR$ in $C^1([0,T])$.  These potentials $V(t,x)$ and $A(t,x)$ do not satisfy either Assumption 2.A, 2.B or 2.C.  We suppose Assumption 2.D, \eqref{1.11} and \eqref{2.16} for $W(t,x), W_s(t,x)$ and $H_s(t,x)$ respectively, where we replace 
\eqref{2.9} with
\begin{equation} \label{2.23}
|\partial_x^{\alpha}W(t,x)| \leq C_{\alpha}<x>^{l_0+1},\quad |\alpha| \geq 1.
\end{equation}
We define $\kdelta(t,0)f$ and $\ksdelta(t,0)f$ by \eqref{2.1} and \eqref{2.15} respectively.  Then we have the same statements (1) - (4) as in Theorems 2.1 and 2.3.  In more general, we suppose about potentials $(V,A)$ that Assumptions 2.1 and 2.2 in \cite{Ichinose 2020} are satisfied.  Then, we can prove the same statements  (1) - (4) as in Theorems 2.1 and 2.3 by using Lemma 5.6, (5.32), (5.37), Theorem 6.5 in \cite{Ichinose 2020} and Theorem 2.1 in \cite{Ichinose 2019}.  We will not give their proofs here because we can prove them by following the proofs of Theorems 2.1 and 2.3 in the present paper, where we use the semi-norms $|\cdot|_{l}$ in $\Sspace$ in place of  $\Vert\cdot\Vert_{a}$ as in the proofs of Theorems 2.1 - 2.4 of \cite{Ichinose 2020}.
\end{rem}
  %
  \section{Stability of $\Cts(\ts)$}
  Let $S(t,s;q)$ and $S_w(t,s;q)$ be the classical actions defined by \eqref{1.3} and \eqref{1.6}, respectively.
  Let $\qts$  be the straight line defined by
\begin{equation}  \label{3.1}
  q^{t,s}_{x,y}(\theta) = y + \frac{\theta-s}{t-s}(x-y), \quad s\leq \theta \leq t
\end{equation}
and write 
\begin{equation} \label{3.2}
\gamma^{t,s}_{x,y}: \gamma^{t,s}_{x,y}(\theta) = (\theta,\qts(\theta)) \in \bR^{d+1}, \  s\leq \theta \leq t.
\end{equation}
Then we have
\begin{align}  \label{3.3}
   S(t,s;q^{t,s}_{x,y})
       & = \frac{m|x - y|^2}{2(t - s)} + \int_{\gamma^{t,s}_{x,y}} \bigl(A\cdot dx - Vdt\bigr)
        \notag \\
     & = \frac{m|x - y|^2}{2(t - s)} + (x - y)\cdot\int^1_0 A(s+ \theta\rho,
y+ \theta (x - y))d\theta \notag \\
   & \qquad -  \int^t_s  V\left(\theta, y+ \frac{\theta - s}{t-s} (x - y)\right)d\theta \notag \\
   & = \frac{m|x - y|^2}{2(t - s)} + (x - y)\cdot\int^1_0 A(t - \theta\rho,
x - \theta (x - y))d\theta \notag \\
   & \qquad - \rho \int^1_0  V(t - \theta\rho,x - \theta (x - y))d\theta, \ \rho = t - s,
\end{align}
\begin{align}  \label{3.4}
   & S_w(t,s;q^{t,s}_{x,y})
        = S(t,s;q^{t,s}_{x,y}) + i\int^t_s  W\left(\theta, y+ \frac{\theta - s}{t-s} (x - y)\right)d\theta \notag \\
       &\qquad = S(t,s;q^{t,s}_{x,y}) + i\rho \int^1_0  W(t - \theta\rho,x - \theta (x - y))d\theta.
\end{align}
Throughout the present paper we often write $\rho = t - s$.
\par
Let $M \geq 0$ be an integer and suppose that $p(x,w) \in \Cspace(\bR^{2d})$ satisfies 
\begin{equation} \label{3.5}
|\partial_{w}^{\alpha}\partial_{x}^{\beta}p(x,w)| \leq C_{\alpha\beta}<x;w>^M, \ (x,w) \in \bR^{2d}
\end{equation}
for all $\alpha$ and $\beta$, where $<x;w> = \sqrt{1 + |x|^2 + |w|^2}$. For $f \in \Czerospace$ we define
\begin{equation}  \label{3.6}
P(t,s)f =
        \begin{cases}
            \begin{split}
              & \sqrt{m/(2\pi i\rho)}^{\ d}
                  \int \bigl(\exp iS_w(t,s; q^{t,s}_{x,y})\bigr) \\
         &\hspace{2cm} \times    p(x,(x-y)/\sqrt{\rho})f(y)dy,\quad  s < t,
                      \end{split}
                  \\
\begin{split}
        & \sqrt{m/(2\pi i)}^{\ d}
    \text{Os}-\int (\exp im|w|^2/2)\\
             &\hspace{2cm}\times p(x,w)dwf(x), \quad s = t.
             \end{split}
        \end{cases}
\end{equation}
 Then the formal adjoint operator $P(t,s)^{\dag}$ of $P(t,s)$ on $\Czerospace$ is given by
\begin{equation*}
P(t,s)^{\dag}f =
        \begin{cases}
            \begin{split}
              & \sqrt{i m/(2\pi \rho)}^{\ d}
                  \int \bigl(\exp iS_w(t,s; q^{t,s}_{y,x})\bigr)^* \\
         &\hspace{2cm} \times    p(y,(y-x)/\sqrt{\rho})^*f(y)dy,\quad  s < t,
                      \end{split}
                  \\
\begin{split}
        & \sqrt{im/(2\pi )}^{\ d}
    \text{Os}-\int (\exp -im|w|^2/2)\\
             &\hspace{2cm}\times p(x,w)^{*}dwf(x),\quad  s = t.
             \end{split}
        \end{cases}
\end{equation*}
We can prove the following as in the proof of Lemma 2.1 of \cite{Ichinose 1999}.
\begin{lem}
Let $p(x,w)$ be a function satisfying \eqref{3.5}.  We assume \eqref{2.7}.  In addition, we assume that $\partial_x^{\alpha}V(t,x)$ and  $\partial_x^{\alpha}A_j(t,x)\ (j = 1,2,\dots,d)$  are continuous in $\domain$ for all $\alpha$ and that there exists a constant $M' \geq 0$ satisfying 
\begin{equation*}
|\partial_x^{\alpha}V(t,x)| + \sum_{j=1}^d|\partial_x^{\alpha}A_j(t,x)| + |\partial_x^{\alpha}W(t,x)| \leq C_{\alpha}<x>^{M'}
\end{equation*}
in $\domain$ for all $\alpha$.
 Then, for $f \in \Sspace$ $\partial_{x}^{\alpha}(P(t,s)f)$ are continuous in $0 \leq s \leq t \leq T$ and $x \in \bR^{d}$ for all $\alpha$.
\end{lem}
 	In particular, when $p(x,w) = 1$, we write $P(t,s)f$ as $\Cts(t,s)f$.  That is, 
\begin{equation}  \label{3.7}
\mathcal{C}(t,s)f =
        \begin{cases}
            \begin{split}
              & \sqrt{m/(2\pi i\rho)}^{\ d}
                  \int \bigl(\exp iS_w(t,s; q^{t,s}_{x,y})\bigr)f(y)dy, \quad s < t, \\
                      \end{split}
                  \\
\begin{split}
        & f,\qquad s = t.
             \end{split}
        \end{cases}
\end{equation}
Then, from \eqref{2.1} we can write 
\begin{equation}  \label{3.8}
     \kdelta(t,0)f = \limepsilon{\cal C}(t,\tau_{\nu-1})\chi(\epsilon\cdot){\cal C}(\tau_{\nu-1},\tau_{\nu-2})\chi(\epsilon\cdot)\cdots \chi(\epsilon\cdot)
{\cal C}(\tau_1,0)f
\end{equation}
for $f \in \Czerospace$.
\par
	For a weight function $W(t,x)$ we set 
\begin{equation}  \label{3.9}
    c_w(t,s;x,y) = \exp \left(-\rho\int_0^1W(t - \theta\rho,x - \theta(x - y))d\theta\right), \ \rho = t - s.
    \end{equation}
\begin{lem}
Let $p(x,w)$ be a function satisfying \eqref{3.5}. We assume that $\partial_x^{\alpha}V(t,x), \partial_x^{\alpha}A_j(t,x)$ and  $\partial_x^{\alpha}\partial_tA_j(t,x)$ are continuous in $\domain$ for  $|\alpha| \leq 1$ and $j = 1,2,\dots,d$. Let $f \in \Czerospace$. Then for any $0 < \epsilon \leq 1$ and $0 \leq s < t \leq T$ we have
\begin{align}  \label{3.10}
    & P(t,s)^{\dag}\chi(\epsilon\cdot)^2P(t,s)f = \left(\frac{m}{2\pi(t-s)}\right)^d\int f(y)dy \int\chi(\epsilon z)^2
    \notag \\
    &\qquad \times  \left(\exp i(x - y)\cdot \frac{m\Phi}{t-s} \right) c_w(t,s;z,x)c_w(t,s;z,y)\notag \\
   &\qquad \times p\left(z,\frac{z - x}{\sqrt{t-s}}\right)^*p\left(z,\frac{z - y}{\sqrt{t-s}}\right)dz,
\end{align}
\begin{equation*} 
  \Phi = \Phi(t,s;x,y,z) = (\Phi_1,\dots,\Phi_d),
\end{equation*}
\begin{align}  \label{3.11}
    & \Phi_j =  z_j - \frac{x_j + y_j}{2} + \frac{t - s }{m}
              \int_0^1A_j(s,x + \theta(y - x))d\theta           \notag \\
      & -\frac{(t - s)^2 }{m}\int_0^1\int_0^1
                \sigma_1E_j(\tau(\sigma),\zeta(\sigma))
                       d\sigma_1d\sigma_2 \notag \\
                        & - \frac{t - s }{m}\sum_{k=1}^d (z_k - 
x_k)\int_0^1\int_0^1\sigma_1B_{jk}
                     (\tau(\sigma),\zeta(\sigma))d\sigma_1d\sigma_2  
\end{align}
or
\begin{align}  \label{3.12}
    & \Phi_j =  z_j - \frac{x_j + y_j}{2} + \frac{t - s }{m}
              \int_0^1A_j(s,x + \theta(y - x))d\theta           \notag \\
      & -\frac{(t - s)^2 }{m}\int_0^1\int_0^1
                \sigma_1E_j(\tau(\sigma),\zeta(\sigma))
                       d\sigma_1d\sigma_2  - \frac{(t - s)^2 }{m}\int_0^1d\theta \sum_{k=1}^d (z_k - 
x_k) \notag \\
                        &\qquad \times \int_0^1\int_0^1\sigma_1(1-\sigma_1)\frac{\partial B_{jk}}{\partial t}
                     (s+\theta(1-\sigma_1)\rho,\zeta(\sigma))d\sigma_1d\sigma_2, 
\end{align}
where
\begin{equation}  \label{3.13}
    \bigl(\tau(\sigma),\zeta(\sigma)\bigr)         
       = \bigl(t - \sigma_1(t - s), z + \sigma_1(x - z)
                  + \sigma_1\sigma_2(y - x)\bigr)\ \in \bR^{d+1}.
\end{equation}
\end{lem}
\begin{proof}
From \eqref{3.4}, \eqref{3.6} and \eqref{3.9}  we can write
\begin{align}  \label{3.14}
    P(t,s)f & = \sqrt{m/(2\pi i\rho)}^{\ d}
                  \int \bigl(\exp iS(t,s; q^{t,s}_{x,y})\bigr)c_w(t,s;x,y) \notag \\
                  & \times p\left(x,\frac{x - y}{\sqrt{t-s}}\right)f(y)dy,\quad \rho = t - s > 0.
\end{align}
Hence, we can easily prove Lemma 3.2 from Lemma 5.2 in \cite{Ichinose 2020}.
\end{proof}
\begin{lem}
We assume that $ \partial_x^{\alpha}A_j(t,x)\ (j = 1,2,\dots,d)$  are continuous for all $\alpha$ and satisfy \eqref{2.5}  in $\domain$.
\\
(1) Suppose that  Assumption 2.A is satisfied.  Let $\Phi_j(t,s;x,y,z)\ (j = 1,2,\dots,d)$ be the functions defined by \eqref{3.11}.  Then, there exist a constant $\rho^* > 0$ such that for all fixed $0 \leq t - s \leq \rho^*$ and $(x,y) \in \bR^{2d}$, the map: $\bR^d \ni z \to \xi = \Phi(\ts;x,y,z) \in \bR^d$ is a homeomorphism, whose inverse will be denoted by the map: $\bR^d \ni \xi \to z = z(\ts;x,\xi,y) \in \bR^d$, and we have 
\begin{equation}   \label{3.15}
           \sum_{j=1}^d |\partial_{\xi}^{\alpha}\partial_x^{\beta}\partial_{y}^{\gamma}
            z_j(t,s;x,\xi,y)|  \leq C_{\alpha\beta\gamma},
\  |\alpha + \beta  + \gamma| \geq 1,
           \end{equation}
\begin{equation}   \label{3.16}
           \det \frac{\partial z}{\partial\xi}(t,s;x,\xi,y) = 1 + (t - s)h(t,s;x,\xi,y) > 0,
           \end{equation}
\begin{equation}   \label{3.17}
            |\partial_{\xi}^{\alpha}\partial_x^{\beta}\partial_{y}^{\gamma}
            h(t,s;x,\xi,y)|  \leq C_{\alpha\beta\gamma} < \infty
           \end{equation}
for all $\alpha, \beta$ and $\gamma$ in $0 \leq t - s \leq \rho^*$ and $(x,\xi,y) \in \bR^{3d}$.
\\
(2)  Suppose that Assumption 2.B is satisfied.  Let  $\Phi_j(t,s;x,y,z)\ (j = 1,2,\dots,d)$ be the functions defined by \eqref{3.12}.  Then we have the same statements as in  (1).
\end{lem}
\begin{proof}
We have already proved (1) in Lemma 3.2, (3.9) and (3.10) of \cite{Ichinose 1999} (cf.  Lemma 3.6, (3.18) and (3.19) of \cite{Ichinose 1997}). 
\par
	We will prove (2).  Let us write the $5$-th term on the right-hand side of \eqref{3.12} as $-(t-s)^2B'(t,s;x,y,z)/m$.  Then, from the assumption \eqref{2.4} we can prove
\begin{equation}   \label{3.17-2}
            |\partial_{x}^{\alpha}\partial_y^{\beta}\partial_{z}^{\gamma}
            B'_j(t,s;x,\xi,y)|  \leq C_{\alpha\beta\gamma}, \quad |\alpha + \beta + \gamma| \geq 1
  \end{equation}
  in $0 \leq s \leq t \leq T$ and $(x,y,z) \in \bR^{3d}$  as in the proof of \eqref{3.15} of \cite{Ichinose 1997}, where $B' = (B'_1,\cdots,B'_d) \in \bR^d$.  Thereby we can prove (2) as in the proof of (1).
\end{proof}
 	From now on we fix $\rho^* > 0$ determined in Lemma 3.3 throughout the present paper.  The following lemma is crucial in the present paper.
	\begin{lem}
	Assume \eqref{2.7} and \eqref{2.8} where we take $C(W) = 0$.  Let $c_w(t,s;x,y)$ be the function defined by \eqref{3.9}.  Then we have
\begin{equation}   \label{3.18}
            0 \leq c_w(t,s;x,y) \leq 1, \ 
           |\partial_x^{\alpha}\partial_y^{\beta}c_w(t,s;x,y)| \leq C_{\alpha\beta}, \ |\alpha + \beta| \geq 1
 \end{equation}
for $0 \leq s \leq t \leq T$ and $(x,y) \in \bR^{2d}$ with constants $C_{\alpha\beta} \geq 0$.
	\end{lem}
\begin{proof}
It is clear from \eqref{2.7} and \eqref{3.9} that the first inequality of \eqref{3.18} holds.  For $a \geq 0$ we can easily see 
\begin{equation}   \label{3.19}
           \sup_{r \geq 0} e^{-r}(T + r)^a = C'_a < \infty
 \end{equation}
with constants $C'_a \geq 0$.  Let $|\alpha| \geq 1$.  Using H\"older's inequality and \eqref{2.8}, we have
\begin{align} \label{3.20}
& \rho \int_0^1|(\partial_x^{\alpha}W)(t - \theta\rho,x - \theta(x-y))|d\theta \notag \\
& \leq \rho \biggl(\int_0^1|(\partial_x^{\alpha}W)(t - \theta\rho,x - \theta(x-y))|^{p_{\alpha}}d\theta \biggr)^{1/p_{\alpha}} \notag \\
& \leq C_{\alpha}^{1/p_{\alpha}}\rho \biggl[\int_0^1\bigl\{1 + W(t - \theta\rho,x - \theta(x-y))\bigr\}d\theta \biggr]^{1/p_{\alpha}} \notag \\
& =  C_{\alpha}^{1/p_{\alpha}}\rho^{1 - 1/p_{\alpha}}\biggl\{\rho + \rho\int_0^1  W(t - \theta\rho,x - \theta(x-y))d\theta \biggr\}^{1/p_{\alpha}} \notag \\
& \leq  C_{\alpha}^{1/p_{\alpha}}\rho^{1 - 1/p_{\alpha}}\biggl\{T + \rho\int_0^1  W(t - \theta\rho,x - \theta(x-y))d\theta \biggr\}^{1/p_{\alpha}}.
\end{align}
Hence, letting $\alpha = (1,0,\dots,0) \in \bR^d$, by \eqref{2.7} and \eqref{3.19} we have
\begin{align} \label{3.21}
& |\partial_{x_1}c_w(t,s;x,y)| \leq \biggl(\exp -\rho\int_0^1W(t - \theta\rho,x - \theta(x-y))d\theta\biggl)
\notag \\
& \times C_{\alpha}^{1/p_{\alpha}}T^{1 - 1/p_{\alpha}}\biggl\{T + \rho\int_0^1  W(t - \theta\rho,x - \theta(x-y))d\theta \biggr\}^{1/p_{\alpha}}  
\notag \\
& \leq C_{\alpha}^{1/p_{\alpha}}T^{1 - 1/p_{\alpha}}C'_{1/p_{\alpha}} < \infty.
\end{align}
In the same way we can complete the proof of  the second inequality of \eqref{3.18}, using \eqref{3.19} and \eqref{3.20}.
\end{proof}
\begin{pro}
Suppose that either Assumption 2.A or 2.B is satisfied.  In addition, we suppose Assumption 2.C, \eqref{2.7} and \eqref{2.8}, where we take $C(W) = 0$ and \eqref{2.6} is replaced with
\begin{equation} \label{3.22}
|\partial_x^{\alpha}V(t,x)| \leq C_{\alpha}<x>^{M_1}, \quad |\alpha| \geq 1
\end{equation}
with an integer $M_1 \geq 1$ independent of $\alpha$.  Let $\rho^* > 0$ be the constant determined in Lemma 3.3 and $\Cts(t,s)$ the operator defined by \eqref{3.7}.  Then there exists a constant $K_0 \geq 0$ such that 
\begin{equation}   \label{3.23}
 \Vert \Cts(\ts)f\Vert \leq e^{K_0(t-s)}\Vert f \Vert, \quad 0 \leq t - s \leq \rho^*
           \end{equation}
           for all $f \in L^2$.
\end{pro}
\begin{proof}
We first suppose that Assumption 2.A is satisfied.  Then, letting $\Phi_j$ be defined by \eqref{3.11}, from \eqref{3.10} we have
\begin{align*}  
    & \Cts(t,s)^{\dag}\chi(\epsilon\cdot)^2\Cts(t,s)f = \left(\frac{m}{2\pi(t-s)}\right)^d\int f(y)dy \int\chi(\epsilon z)^2
    \notag \\
    &\qquad \times  \left(\exp i(x - y)\cdot \frac{m\Phi}{t-s} \right) c_w(t,s;z,x)c_w(t,s;z,y)dz 
  \end{align*}
  for $f \in \Sspace$.  We will use (1) in Lemma 3.3.  Letting $0 \leq t - s \leq \rho^*$ and making the change of variables: $\bR^d \ni z \to \xi = \Phi(t,s;x,y,z) \in \bR^d$ in the above equation, we have
\begin{align*}  
    & \Cts(t,s)^{\dag}\chi(\epsilon\cdot)^2\Cts(t,s)f = \left(\frac{m}{2\pi(t-s)}\right)^d\int f(y)dy \int\chi(\epsilon z)^2
    \notag \\
    & \times  \left(\exp i(x - y)\cdot \frac{m\xi}{t-s} \right) c_w(t,s;z,x)c_w(t,s;z,y)
    \det\frac{\partial z}{\partial \xi}(t,s;x,\xi,y)d\xi 
  \end{align*}
  with $z = z(t,s;x,\xi,y)$, which shows
  \begin{align} \label{3.24} 
    & \Cts(t,s)^{\dag}\chi(\epsilon\cdot)^2\Cts(t,s)f = \left(\frac{1}{2\pi}\right)^d\int e^{i(x-y)\cdot\eta}d\eta\int\chi(\epsilon z)^2c_w(t,s;z,x)
       \notag \\
    & \ \times c_w(t,s;z,y)\det\frac{\partial z}{\partial \xi}(t,s;x,(t-s)\eta/m,y) f(y)dy,\  0 \leq t - s \leq \rho^*
  \end{align}
  with $z = z(t,s;x,(t-s)\eta/m,y)$.  Hence, noting \eqref{3.15} and \eqref{3.18}, we can easily prove
  \begin{align} \label{3.25} 
    &  \limepsilon\Cts(t,s)^{\dag}\chi(\epsilon\cdot)^2\Cts(t,s)f\notag \\
    & = \left(\frac{1}{2\pi}\right)^d\int e^{i(x-y)\cdot\eta}d\eta\int c_w(t,s;z,x)
    c_w(t,s;z,y)
       \notag \\
    & \ \times \det\frac{\partial z}{\partial \xi}(t,s;x,(t-s)\eta/m,y) f(y)dy,\  0 \leq t - s \leq \rho^*
  \end{align}
 with $z = z(t,s;x,(t-s)\eta/m,y)$ in the topology of $\Sspace$, which we write as $\Cts(t,s)^{\dag}\,\Cts(t,s)f$ formally.
 \par
 	Noting \eqref{3.15} - \eqref{3.17} and \eqref{3.18}, let us apply Theorem 1.A in the introduction to \eqref{3.25}.
Then we have
\begin{align}  \label{3.26}
     \Vert \Cts(\ts)^{\dag}\,\Cts(\ts)f\Vert &\leq \bigl\{1 + 2K_0(t-s)\bigr\}\Vert f\Vert \notag \\
    & \leq e^{2K_0(t-s)} \Vert f\Vert, \quad  0 \leq t - s \leq \rho^*
\end{align}
with a constant $K_0 \geq 0$.  Consequently from \eqref{3.25} we have
\begin{align*}  
\Vert\Cts(t,s)f\Vert^2 &\leq \lim_{\epsilon \to 0+}\bigl(\chi(\epsilon\cdot)\Cts(\ts)f,\chi(\epsilon\cdot)\Cts(\ts)f\bigr) \\
    & = \lim_{\epsilon \to 0+}\bigl(\Cts(\ts)^{\dag}\chi(\epsilon\cdot)^2\Cts(\ts)f,f\bigr)  \\
    & = \bigl(\Cts(\ts)^{\dag}\,\Cts(\ts)f,f\bigr) \leq e^{2K_0(t-s)} \Vert f\Vert^2
\end{align*}
for $f \in \Sspace$ by using Fatou's lemma, which shows \eqref{3.23}.
\par
	Next we suppose that Assumption 2.B is satisfied.  Letting $\Phi_j$ be defined by \eqref{3.12}, we can also prove \eqref{3.25} and \eqref{3.26} as in the proof of the first case.  Consequently we can prove \eqref{3.23}.
\end{proof}
\begin{pro}
Let $p(x,w)$ be a function satisfying \eqref{3.5} and $P(t,s)$ the operator defined by \eqref{3.6}.  Then, under the assumptions of Proposition 3.5 we have
\begin{equation} \label{3.27} 
\Vert P(t,s)f\Vert_a \leq 
     C_a \Vert f\Vert_{M+aM_1}, \quad  0 \leq t - s \leq \rho^*
\end{equation}
for $a = 0,1,2,\dots$ and all $f \in B^{M+aM_1}$ with constants $C_a \geq 0$, where $M_1$ is the integer in \eqref{3.22}.
\end{pro}
\begin{proof}
Setting 
\begin{equation} \label{3.28} 
p'(t,s;x,w) := p(x,w)c_w(t,s;x,x - \sqrt{\rho}w),
\end{equation}
from \eqref{3.14} we have
\begin{align}  \label{3.29}
    P(t,s)f & = \sqrt{m/(2\pi i\rho)}^{\ d}
                  \int \bigl(\exp iS(t,s; q^{t,s}_{x,y})\bigr)p'\left(t,s;x,\frac{x - y}{\sqrt{t-s}}\right)f(y)dy, \notag \\
                 &\quad \rho = t - s > 0
\end{align}
and also from \eqref{3.18}
\begin{equation} \label{3.30} 
|\partial_w^{\alpha}\partial_x^{\beta}p'(t,s;x,w)| \leq C_{\alpha\beta}<x;w>^M
\end{equation}
for all $\alpha$ and $\beta$.
\par
	At first we suppose that Assumption 2.A is satisfied. Then, using \eqref{3.29} and \eqref{3.30}, we can prove \eqref{3.27}  from Theorem 4.4 of \cite{Ichinose 1999}.  We can also prove \eqref{3.27} under Assumption 2.B, noting \eqref{3.17-2} and following the proof of Theorem 4.4 in \cite{Ichinose 1999}.
\end{proof}
\section{Consistency of $\Cts(\ts)$}
\begin{lem}
Let $H_w(t)$ be the operator defined by \eqref{1.8}.  We assume that  for  all $\alpha$ $\partial_x^{\alpha}V(t,x), \partial_x^{\alpha}A_j(t,x)\ (j = 1,2,\dots,d)$ and $ \partial_x^{\alpha}\partial_tA_j(t,x)$  are continuous in $\domain$ and satisfy
\begin{align*}
&|\partial_x^{\alpha}V(t,x)| + \sum_{j=1}^d\bigl(|\partial_x^{\alpha}A_j(t,x)| + |\partial_x^{\alpha}\partial_tA_j(t,x)|  \bigr)+ |\partial_x^{\alpha}W(t,x)|
\notag \\
& \leq C_{\alpha}<x>^{M'}
\end{align*}
with constants $C_{\alpha} \geq 0$ and $M' \geq 0$, where $M'$ is independent of $\alpha$.  Then, there exists a function $r(t,s;x,w)$ satisfying \eqref{3.5} for an integer $M \geq 0$ such that $\partial_w^{\alpha}\partial_x^{\beta}r(t,s;x,w)$ are continuous in $0 \leq s \leq t \leq T$ and $(x,w) \in \bR^{2d}$ for all $\alpha, \beta$ and we have
\begin{equation}  \label{3.31}
      \left\{i\frac{\partial}{\partial t} - H_w(t)\right\}\Cts(\ts) f
                        = \sqrt{t - s} R(\ts)f
\end{equation}
for $f \in \Czerospace$.
\end{lem}
\begin{proof}
We note \eqref{1.6} and \eqref{3.7}.  Then, replacing $V(t,x)$ with $V(t,x) - iW(t,x)$ in the proof of Lemma 4.1 of \cite{Ichinose 2003}, we can complete the proof of Lemma 4.1.
\end{proof}
\begin{pro}
Besides the assumptions of Proposition 3.5 we assume
\begin{equation*}
 |\partial_x^{\alpha}W(t,x)|
 \leq C_{\alpha}<x>^{M''}
\end{equation*}
in $\domain$ for all $\alpha$ with constants $C_{\alpha} \geq 0$ and $M'' \geq 0$, where $M''$ is independent of $\alpha$.  Then, there exists a function $r(t,s;x,w)$ satisfying the properties stated in Lemma 4.1 and we have
\begin{equation} \label{3.32} 
\Vert R(t,s)f\Vert_a \leq 
     C_a \Vert f\Vert_{M+aM_1}, \quad  0 \leq t - s \leq \rho^*
\end{equation}
for $a = 0,1,2,\dots$ and all $f \in B^{M+aM_1}$, where $M_1$ is the integer in \eqref{3.22}.
\end{pro}
\begin{proof}
From \eqref{1.1} we have $\partial_tA_j = -E_j - \partial_{x_j}V$.  Hence we see that $\partial_x^{\alpha}\partial_tA_j(t,x)$ are continuous in $\domain$ for all $\alpha$ from the assumptions.  In addition, from \eqref{2.2} and \eqref{3.22} we have
\begin{equation*}
 |\partial_x^{\alpha}\partial_tA_j(t,x)| \leq |\partial_x^{\alpha}E_j(t,x)| + |\partial_x^{\alpha}\partial_{x_j}V(t,x)|
 \leq C_{\alpha}<x>^{M_1}, \ |\alpha| \geq 1.
\end{equation*}
Hence Lemma 4.1 holds.  Applying Proposition 3.6 to $R(t,s)f$, we get \eqref{3.32}.
\end{proof}
	\vspace{0.3cm}
	Making the change of variables: $\bR^d \ni y \to w = (x-y)/\sqrt{\rho} \in \bR^d$ in \eqref{3.6}, from \eqref{3.3} and \eqref{3.4} we have
\begin{equation}  \label{3.33}
    P(t,s)f  = \sqrt{\frac{m}{2\pi i}}^{\, d}
                  \int e^{i\phi(t,s;x,w)}p(x,w)f(x-\sqrt{\rho}w)dw,\  \rho = t - s > 0
\end{equation}
for $f \in \Czerospace$ as in the proof of (2.9) of \cite{Ichinose 1999}, where	
\begin{align}  \label{3.34}
   & \phi(t,s;x,w)  = \frac{m}{2}|w|^2 + \sqrt{\rho}w\cdot\int_0^1A(t-\theta\rho,x-\theta\sqrt{\rho}w)d\theta
   \notag \\
   & -\rho\int_0^1V(t-\theta\rho,x-\theta\sqrt{\rho}w)d\theta + i\rho\int_0^1W(t-\theta\rho,x-\theta\sqrt{\rho}w)d\theta.
                  \end{align}
\begin{lem}
Suppose Assumption 2.C and \eqref{2.9}. Let $\Cts(t,s)$ be the operator defined by \eqref{3.7}.  Then, for an arbitrary multi-index $\kappa$ both of commutators $[\partial_x^{\kappa},\Cts(t,s)]f$ and $[x^{\kappa},\Cts(t,s)]f$ for $f \in \Czerospace$ are written in the form
\begin{align}  \label{3.35}
   & (t-s)\sum_{|\gamma| < |\kappa|}\widetilde{P}_{\gamma}(t,s)(\partial_x^{\gamma}f) := (t-s) \sum_{|\gamma| < |\kappa|}   \sqrt{\frac{m}{2\pi i}}^{\, d}\notag \\
   & \times \int e^{i\phi(t,s;x,w)}p_{\gamma}(t,s;x,\sqrt{\rho}w)(\partial_x^{\gamma}f)(x-\sqrt{\rho}w)dw,             
   \end{align}
   where $p_{\gamma}(t,s;x,\zeta)$ satisfy
\begin{equation}  \label{3.36}
   |\partial_{\zeta}^{\alpha}\partial_x^{\beta}p_{\gamma}(t,s;x,\zeta)| \leq C_{\alpha\beta}<x;\zeta>^{|\kappa| - |\gamma|}
\end{equation}
for all $\alpha$ and $\beta$.
\end{lem}
\begin{proof}
Replacing $V(t,x)$ with $V(t,x) - iW(t,x)$ in the proof of Lemma 3.2 of \cite{Ichinose 2003}, we can prove Lemma 4.3.
\end{proof}
\begin{pro} 
Suppose that the assumptions of Theorem 2.2 are satisfied, where we take $C(W) = 0$.  Then, for $a = 0,1,2,\dots$ there exist constants $K_a \geq 0$ such that 
\begin{equation}   \label{3.37}
 \Vert \Cts(\ts)f\Vert_a \leq e^{K_a(t-s)}\Vert f \Vert_a, \quad 0 \leq t - s \leq \rho^*
           \end{equation}
for all $f \in B^a$.
\end{pro}
 \begin{proof}
Let $|\kappa| = a$. Using Proposition 3.6 and Lemma 4.3, we have
\begin{align*}
&\Vert x^{\kappa}(\Cts(\ts) f)\Vert 
\leq \Vert\Cts(\ts)(x^{\kappa}f)\Vert +  
(t - s)\sum_{|\gamma| < a}\Vert \tilde{P}_{\gamma}(t,s)(\partial_x^{\gamma}f)\Vert \\
& \leq \Vert\Cts(\ts)(x^{\kappa}f)\Vert +  
C(t - s)\sum_{|\gamma| < a}\Vert \partial_x^{\gamma}f\Vert_{a-|\gamma|} \\
& \leq \Vert\Cts(\ts)(x^{\kappa}f)\Vert +  
C'(t - s)\Vert f\Vert_{a}.
\end{align*}
Here we used $\Vert \partial_x^{\gamma}f\Vert_{a-|\gamma|} \leq \text{Const.}\Vert f \Vert_{a}$ from (4.21) in \cite{Ichinose 1999}. Hence from Proposition 3.5 we have
\begin{equation} \label{3.38}
\Vert x^{\kappa}(\Cts(\ts) f)\Vert \leq e^{K_{0}(t -s)}\Vert x^{\kappa}f\Vert +
   C'(t - s)\Vert f\Vert_{a}.
\end{equation}
In the same way we have
\begin{equation} \label{3.39}
\Vert \partial_{x}^{\kappa}(\Cts(\ts) f)\Vert \leq e^{K_{0}(t -s)}\Vert \partial_{x}^{\kappa}f\Vert +
   C''(t - s)\Vert f\Vert_{a}.
\end{equation}
Since $\Vert f \Vert_{a}$ is defined by \eqref{2.12}, from \eqref{3.23}, \eqref{3.38} and \eqref{3.39} we obtain 
\begin{align*}
&\Vert \Cts(\ts) f\Vert_{a}
= \Vert\Cts(\ts)f\Vert +  \sum_{|\kappa| = a}\Bigl(\Vert x^{\kappa}(\Cts(\ts)f)\Vert + \Vert \partial_{x}^{\kappa}(\Cts(\ts)f)\Vert\Bigr)\\
& \leq e^{K_0(t-s)}\Vert f\Vert_{a} +K'_{0}(t-s)\Vert f\Vert_{a} = \Bigl(e^{K_0(t-s)} + K'_{0}(t-s)\Bigr)\Vert f\Vert_{a} \\
& \leq e^{(K_0+K'_{0})(t-s)}\Vert f\Vert_{a},
\end{align*}
which shows \eqref{3.37}.
 \end{proof}
\begin{thm}  Suppose Assumption 2.C and \eqref{2.9}.  Then for any $u_0 \in B^a\ (a = 0, \pm1,\pm2,\dots)$
 there exists the unique solution $u(t)$ in $C^0_t([0,T];B^a) \cap C^1_t([0,T];B^{a-2})$  with $u(0) = u_0$ to the equation \eqref{1.8}. This solution $u(t)$ satisfies
 \begin{equation} \label{3.40}
 \Vert u(t)\Vert_a\leq  C_a\Vert u_0\Vert_a,  \ 0 \leq t \leq T.
 \end{equation}
 \end{thm}
 \begin{proof}
 The results corresponding to Theorem 4.5 have been proved  in (1) of Theorem 2.1 of \cite{Ichinose 2019}, where  $\|f\|_a$ was defined by $ \|f\| + \sum_{|\alpha| =  2a}\bigl(\|x^{\alpha}f\| +\|\partial_x^{\alpha}f\|\bigr)$ differently from \eqref{2.12}.  Following the proof of (1) of Theorem 2.1 of  \cite{Ichinose 2019}, we can prove Theorem 4.5 as below.  We set $\chi_{\epsilon}(x,\xi) := \chi\bigl(\epsilon(<x> + <\xi>)\bigr)\ (0 < \epsilon \leq 1)$ and $\lambda(x,\xi) := \mu + <x> + <\xi>$, where $\mu > 0$ is the constant such that there exist a function $w(x,\xi)$   satisfying
  \begin{equation*} 
W(X,D_{x})f = (\mu + <X> + <D_{x}>)^{-1}f
  \end{equation*}
  for $f \in \Sspace$ and 
  \begin{equation*} 
 |\partial^{\alpha}_{\xi}\partial^{\beta}_{x}w(x,\xi)| \leq C_{\alpha\beta}(1 + <x> + <\xi>)^{-1}
  \end{equation*}
  for all $\alpha$ and $\beta$ (cf. Lemma 2.3 of \cite{Ichinose 1995}). 
  \par
  We set 
 \begin{align*} 
Q_{\epsilon}(X,D_{x}) = &\Bigl[\Lambda(X,D_{x}), X_{\epsilon}(X,D_{x})^{\dag}H_{w}(t) X_{\epsilon}(X,D_{x})\Bigr]\Lambda(X,D_{x})^{-1} 
   \end{align*}
 as in (4.3) of \cite{Ichinose 2019}.
 Then we can prove
  \begin{equation*} 
 |\partial^{\alpha}_{\xi}\partial^{\beta}_{x}q_{\epsilon}(x,\xi)| \leq C_{\alpha\beta} < \infty
  \end{equation*}
   for all $\alpha$ and $\beta$ with constants $C_{\alpha\beta}$ independent of $0 < \epsilon \leq 1$ as in the proof of Lemma 4.1 of \cite{Ichinose 2019} and Lemma 3.1 of \cite{Ichinose 1995}.  Therefore we obtain the results corresponding to Lemma 4.1 of \cite{Ichinose 2019}.  Then we can complete the remaining proof of Theorem 4.5, following the proof of Theorem 2.1 of \cite{Ichinose 2019}.
 \end{proof}
 \begin{pro} 
 We suppose the same assumptions as in Proposition 4.4.  Let $U(t,s)f$ be the solution to \eqref{1.8} found in Theorem 4.5.  Then there exists an integer $M \geq 0$ such that we have
 \begin{equation} \label{3.41}
 \Vert \Cts(\ts)f - U(t,s)f\Vert_a \leq  C_a\rho^{3/2}\Vert f\Vert_{M+a},\quad  0\leq  t-s \leq \rho^*
 \end{equation}
 for $a = 0,1,2,\dots.$
 \end{pro}
 \begin{proof}
 Using \eqref{3.31}, we can write
\begin{align*}
& i\bigl\{\Cts(\ts)f - f\bigr\} = i\bigl\{\Cts(s+\rho,s)f - f\bigr\} = i\rho\int_0^1\frac{\partial\Cts}{\partial t}(s+\theta\rho,s)fd\theta \\
& = \rho\int_0^1\Bigl\{H_w(s+\theta\rho)\Cts(s+\theta\rho,s)f + \sqrt{\theta\rho}R(s+\theta\rho,s)f\Bigr\}d\theta
\end{align*}
and so
\begin{align*}
& i\frac{\Cts(\ts)f - f}{\rho} - H_w(s) f = \sqrt{\rho}\int_0^1\sqrt{\theta}R(s+\theta\rho,s)fd\theta + \int_0^1H_w(s+\theta\rho)\\
& \quad \cdot\big\{\Cts(s+\theta\rho,s)f - f\bigr\}d\theta + \int_0^1\bigl\{H_w(s+\theta\rho)f-H_w(s)f\big\}d\theta.
\end{align*}
Using 
\begin{align*}
& \Cts(s+\theta\rho,s)f - f = \theta\rho\int_0^1\frac{\partial\Cts}{\partial t}(s+\theta'\theta\rho,s)fd\theta'\\
& = \frac{\theta\rho}{i}\int_0^1\bigl\{H_w(s+\theta'\theta\rho)\Cts(s+\theta'\theta\rho,s)f 
+ \sqrt{\theta'\theta\rho}R(s+\theta'\theta\rho,s)f\bigr\}d\theta',
\end{align*}
we have
\begin{align} \label{3.42}
&  i\frac{\Cts(\ts)f - f}{\rho} - H_w(s) f = \sqrt{\rho}\int_0^1\sqrt{\theta}R(s+\theta\rho,s)fd\theta + \frac{\rho}{i}\int_0^1\theta H_w(s+\theta\rho)d\theta \notag \\
& \quad \cdot\int_0^1\bigl\{H_w(s+\theta'\theta\rho)\Cts(s+\theta'\theta\rho,s)f 
+ \sqrt{\theta'\theta\rho}R(s+\theta'\theta\rho,s)f\bigr\}d\theta'\notag \\
& \quad + \int_0^1\bigl\{H_w(s+\theta\rho)f-H_w(s)f\big\}d\theta.
\end{align}
In the same way 
\begin{align} \label{3.43}
&  i\frac{U(\ts)f - f}{\rho} - H_w(s) f =  \frac{\rho}{i}\int_0^1\theta H_w(s+\theta\rho)d\theta \notag \\
& \quad \cdot\int_0^1\bigl\{H_w(s+\theta'\theta\rho)U(s+\theta'\theta\rho,s)f \bigr\}d\theta'
\notag \\
& \qquad + \int_0^1\bigl\{H_w(s+\theta\rho)f-H_w(s)f\big\}d\theta.
\end{align}
Taking difference between \eqref{3.42} and \eqref{3.43}, we have
\begin{align} \label{3.44}
&  i\left\{\Cts(\ts)f - U(t,s)f\right\} = \rho^{3/2}\int_0^1\sqrt{\theta}R(s+\theta\rho,s)fd\theta + \frac{\rho^2}{i}\int_0^1\theta H_w(s+\theta\rho)d\theta \notag \\
& \quad \cdot\int_0^1\bigl\{H_w(s+\theta'\theta\rho)\Cts(s+\theta'\theta\rho,s)f 
+ \sqrt{\theta'\theta\rho}R(s+\theta'\theta\rho,s)f\bigr\}d\theta'\notag \\
& \quad - \frac{\rho^2}{i}\int_0^1\theta H_w(s+\theta\rho)d\theta \int_0^1H_w(s+\theta'\theta\rho)U(s+\theta'\theta\rho,s)f d\theta'.
\end{align}
Consequently, noting \eqref{2.5}, \eqref{2.6} and \eqref{2.9}, and applying \eqref{3.32} with $M_1 = 1$, \eqref{3.37} and \eqref{3.40} to \eqref{3.44}, we obtain
\begin{align} \label{3.45}
&  \Vert\Cts(\ts)f - U(t,s)f \Vert_a \leq C_1 \rho^{3/2}\Vert f\Vert_{M+a} + C_2 \rho^{2}
(\Vert f\Vert_{4+a} + \Vert f\Vert_{M+2+a}) \notag\\
& \quad + C_3 \rho^{2}\Vert f\Vert_{4+a}
\end{align}
with constants $C_j\ (j=1,2,3)$, which shows \eqref{3.41}.
 \end{proof}
 %
\section{Proofs of Theorems 2.1 and 2.2}
\begin{lem}
We suppose the same assumptions as in Proposition 3.5.  Let $\kdelta(t,0)f$ and $\Cts(\ts)f$ be the operators defined by \eqref{2.1} and \eqref{3.7} respectively. Then we have
\begin{equation} \label{4'.1}
\kdelta(t,0)f = \Cts(t,\tau_{\nu -1})\Cts(\tau_{\nu-1},\tau_{\nu -2})\cdots \Cts(\tau_1,0)f
\end{equation}
for all $f \in L^2$ and all $\Delta$ such that $|\Delta| \leq \rho^*$.
\end{lem}
\begin{proof}
From \eqref{3.8} we could write
\begin{equation*} 
\kdelta(t,0)f = \lim_{\epsilon \to 0}\Cts(t,\tau_{\nu -1})\chi(\epsilon\cdot)\Cts(\tau_{\nu-1},\tau_{\nu -2})\chi(\epsilon\cdot)\cdots \chi(\epsilon\cdot)\Cts(\tau_1,0)f
\end{equation*}
for $f \in \Czerospace$.  Then from \eqref{3.23} we have
\begin{align} \label{4'.2}
& \Vert \Cts(t,\tau_{\nu -1})\chi(\epsilon\cdot)\Cts(\tau_{\nu-1},\tau_{\nu -2})\chi(\epsilon\cdot)\cdots \chi(\epsilon\cdot)\Cts(\tau_1,0)f \notag\\
& \quad  - \Cts(t,\tau_{\nu -1})\Cts(\tau_{\nu-1},\tau_{\nu -2})\cdots \Cts(\tau_1,0)f\Vert \notag\\
& = \biggl\Vert \sum_{j=1}^{\nu-1} \Cts(t,\tau_{\nu -1})\chi(\epsilon\cdot)\Cts(\tau_{\nu-1},\tau_{\nu -2})\chi(\epsilon\cdot)\cdots \chi(\epsilon\cdot)\Cts(\tau_{j+1},\tau_j)\notag \\
& \quad \cdot\{\chi(\epsilon\cdot)- 1\}\Cts(\tau_{j},\tau_{j-1})\Cts(\tau_{j-1},\tau_{j-2})\cdots \Cts(\tau_1,0)f\biggr\Vert \notag\\
& \leq C\sum_{j=1}^{\nu-1}\Vert \{\chi(\epsilon\cdot)- 1\}\Cts(\tau_{j},\tau_{j-1})\Cts(\tau_{j-1},\tau_{j-2})\cdots \Cts(\tau_1,0)f \Vert
\end{align}
for $f \in L^2$ with a constant $C \geq 0$ independent of $0 < \epsilon \leq 1$, which shows \eqref{4'.1} for $f \in L^2$.
\end{proof}
	Now we will prove Theorems 2.1 and 2.2. We can easily see that  we may assume $C(W) = 0$ in Assumption 2.D without loss of generality, because we have only to take $W(t,x) + C(W)$ in place of $W(t,x)$ in \eqref{1.8} and \eqref{2.1}.
	  Hence we assume $C(W) = 0$ hereafter in this section.\par
	We will first prove Theorem 2.2. Using \eqref{4'.1}, we write 
\begin{align} \label{4'.3}
& \kdelta(t,0)f - U(t,0)f  = \Cts(t,\tau_{\nu -1})\Cts(\tau_{\nu-1},\tau_{\nu -2})\cdots\Cts(\tau_1,0)f - U(t,\tau_{\nu -1})
\notag \\
&\quad \cdot U(\tau_{\nu-1},\tau_{\nu -2})\cdots U(\tau_1,0)f = \sum_{j=1}^{\nu} \Cts(t,\tau_{\nu -1})\Cts(\tau_{\nu-1},\tau_{\nu -2})\cdot\cdots\Cts(\tau_{j+1},\tau_j) \notag \\
& \quad \cdot\{\Cts(\tau_{j},\tau_{j-1})- U(\tau_{j},\tau_{j-1})\}U(\tau_{j-1},\tau_{j-2})\cdots U(\tau_1,0)f
= \sum_{j=1}^{\nu} \Cts(t,\tau_{\nu -1})
\notag \\
& \cdot\Cts(\tau_{\nu-1},\tau_{\nu -2})\cdots\Cts(\tau_{j+1},\tau_j)\{\Cts(\tau_{j},\tau_{j-1})- U(\tau_{j},\tau_{j-1})\}U(\tau_{j-1},0)f.
\end{align}
Hence, using \eqref{3.37}, \eqref{3.40} and \eqref{3.41}, we have
\begin{align} \label{4'.4}
& \Vert\kdelta(t,0)f - U(t,0)f\Vert_{a}  \leq  \sum_{j=1}^{\nu} C_{a}e^{K_{a}t} (\tau_{j}-\tau_{j-1})^{3/2}\Vert U(\tau_{j-1},0)f\Vert_{M+a} \notag \\
& \leq C'_{a}\sqrt{|\Delta|}e^{K_{a}T}T\Vert f\Vert_{M+a}.
\end{align}
\par
Let $f \in B^{a}$ be arbitrary.  For any $\epsilon > 0$ we take a $g \in B^{a+M}$ such that 
\begin{equation} \label{4'.5}
\Vert g - f \Vert_{a} < \epsilon.
\end{equation}
Using \eqref{3.37}, \eqref{3.40} and \eqref{4'.5}, we have
\begin{align} \label{4'.6}
& \Vert\kdelta(t,0)f - U(t,0)f\Vert_{a}  \leq  \Vert\kdelta(t,0)g - U(t,0)g\Vert_{a} \notag \\
& \quad + \Vert\kdelta(t,0)(f-g)\Vert_{a} + \Vert U(t,0)(f-g)\Vert_{a} \notag \\
& \leq C'_{a}\sqrt{|\Delta|}e^{K_{a}T}T\Vert g\Vert_{M+a} + \bigl(e^{K_{a}T} + C_{a}\bigr)\epsilon.
\end{align}
Hence 
\begin{equation*} 
\overline{\lim}_{|\Delta|\to 0}\Vert\kdelta(t,0)f - U(t,0)f\Vert_{a} \leq \bigl(e^{K_{a}T} + C_{a}\bigr)\epsilon,
\end{equation*}
which shows 
\begin{equation} \label{4'.7}
\lim_{|\Delta|\to 0}\Vert\kdelta(t,0)f - U(t,0)f\Vert_{a} = 0.
\end{equation}
\par
	Now consider the gauge transformation \eqref{2.10}.  From \eqref{1.2}, \eqref{1.5} and \eqref{1.6} we can easily see
\begin{equation} \label{4'.8}
S'_{w}(t,s;\qts) = S_{w}(t,s;\qts) + \psi(t,x) - \psi(s,y)
\end{equation}
(cf. p. 1024 in \cite{Ichinose 1999}), which shows
\begin{equation} \label{4'.9}
\Cts'(\ts)f = e^{i\psi(t,\cdot)}\Cts(\ts)e^{-i\psi(s,\cdot)}f.
\end{equation}
Consequently we can prove \eqref{2.11} from \eqref{4'.1}.  Thus we could complete the proof of Theorem 2.2.
\par
	Next we will prove Theorem 2.1 by using Theorem 2.2, where we will use only the results in $L^{2}$. We are supposing Assumption 2.A.  Consequently, using Lemma 6.1 in \cite{Ichinose 1999}, we can find a potential $(V', A')$ satisfying \eqref{2.5} and \eqref{2.6}.  From Theorem 2.2 we have \eqref{4'.7} with $a = 0$ for $\kdelta'(t,0)f$ and $U'(t,0)f$ with this potential $(V', A')$.  Let $(V, A)$ be an arbitrary potential stated in Theorem 2.1.  Then, from the proof of Theorem in \cite{Ichinose 1999} on p.1023 we can find a real-valued function $\psi(t,x)$ with continuous $\partial_{x_{j}}\partial_{x_{k}}\psi$ and $\partial_{t}\partial_{x_{j}}\psi\ (j,k = 1,2,\dots,d)$ in $\domain$ satisfying \eqref{2.10}.  Then from Theorem 2.2 we have
\begin{equation} \label{4'.10}
\kdelta(t,0)f = e^{-i\psi(t,\cdot)}\kdelta'(t,0)e^{i\psi(s,\cdot)}f,
\end{equation}
which shows 
\begin{equation} \label{4'.11}
\lim_{|\Delta|\to 0}\kdelta(t,0)f = e^{-i\psi(t,\cdot)}U'(t,0)e^{i\psi(s,\cdot)}f = U(t,0)f \quad \text{in}\ L^{2}
\end{equation}
 for $f \in L^{2}$ because of $ U(t,0)f =  e^{-i\psi(t,\cdot)}U'(t,0)e^{i\psi(s,\cdot)}f $.  
 \par
 	Next consider the gauge transformation 
\begin{equation*}  
    V'' = V -\frac{\partial\varphi}{\partial  t}, \quad  A''_j = A_j + \frac{\partial\varphi}{\partial  x_j}\quad (j= 1,2,\dots,d)
\end{equation*}
stated in Theorem 2.1.  Then we have
\begin{equation*}  
    V' = V'' -\frac{\partial (\psi -\varphi)}{\partial  t}, \quad  A'_j = A''_j + \frac{\partial(\psi -\varphi)}{\partial  x_j}
\end{equation*}
together with \eqref{2.10}. Hence from \eqref{4'.10} we have
\begin{align*} 
\kdelta''(t,0)f & = e^{-i\psi(t,\cdot)+i\varphi(t,\cdot)}\kdelta'(t,0)e^{i\psi(s,\cdot)-i\varphi(s,\cdot)}f \\
& = e^{i\varphi(t,\cdot)}\kdelta(t,0)e^{-i\varphi(s,\cdot)}f,
\end{align*}
which shows \eqref{2.11}.  Thus we could complete the proof of Theorem 2.1.
%

%
%
%
%
\section{Proofs of Theorems 2.3 - 2.6}
	Let $C(W)$ be the constant in Assumption 2.D and $W_s(t,x)$ the Hermitian matrix in Theorems 2.3 and 2.4. As in the proofs of Theorems 2.1 and 2.2 we may assume
\begin{equation}   \label{4.1}
 C(W) = 0,\quad W_s(t,x) \geq 0
 \end{equation}
 in the proofs of Theorems 2.3 and 2.4, because we are assuming \eqref{1.11}.
 \par
 	Using $\mathcal{F}(t,s;\qts)$ defined by the solution to \eqref{2.13}, we define 
\begin{equation}  \label{4.2}
\mathcal{C}_s(t,s)f =
        \begin{cases}
            \begin{split}
              & \sqrt{m/(2\pi i\rho)}^{\ d}
                  \int \bigl(\exp iS_w(t,s; q^{t,s}_{x,y})\bigr) \\
         &\quad \times    \mathcal{F}(t,s;\qts)f(y)dy,\quad  s < t,
                      \end{split}
                  \\
         f, \quad s = t
        \end{cases}
\end{equation}
for $f \in \Czerospace^l$, which is corresponding to $\mathcal{C}(t,s)$ defined by \eqref{3.7}.  Then we can write $\ksdelta(t,0)f$ defined by \eqref{2.15} as
\begin{equation}  \label{4.3}
     \ksdelta(t,0)f = \limepsilon{\cal C}_s(t,\tau_{\nu-1})\chi(\epsilon\cdot){\cal C}_s(\tau_{\nu-1},\tau_{\nu-2})\chi(\epsilon\cdot)\cdots \chi(\epsilon\cdot)
{\cal C}_s(\tau_1,0)f
\end{equation}
for $f \in \Czerospace^l$ in the same way as we did \eqref{3.8}, using 
\begin{align} \label{4.4}
& \mathcal{F}(t,0;\qdelta) \notag \\ 
& = \mathcal{F}\bigl(t,\tau_{\nu-1};q^{t,\tau_{\nu-1}}_{x,x^{(\nu-1)}}\bigr)\mathcal{F}\bigl(\tau_{\nu-1},\tau_{\nu-2};q^{\tau_{\nu-1},\tau_{\nu-2}}_{x^{(\nu-1)},x^{(\nu-2)}}\bigr) \cdots 
  \mathcal{F}\bigl(\tau_{1},0;q^{\tau_{1},0}_{x^{(1)},x^{(0)}}\bigr)
\end{align}
which has been easily proved in Lemma 2.1 of \cite{Ichinose 2007}.
\begin{lem}
(1)  Assume $W_s(t,x) \geq 0$ in $\domain$.  Let $q(\theta) \in \bR^d\ (s \leq \theta \leq t)$ be a continuous path.  Then we have
\begin{equation} \label{4.5}
0 \leq \mathcal{F}(t,s;q)^{\dag}\mathcal{F}(t,s;q) \leq 1,
\end{equation}
\begin{equation} \label{4.6}
\sum_{i=1}^l |\mathcal{F}_{ij}(t,s;q)|^2 \leq 1,\ j = 1,2,\dots,l,
\end{equation}
where $\mathcal{F}_{ij}(t,s;q)$ denotes the $(i,j)$-component of $\mathcal{F}(t,s;q)$. (2) Assume $W_s(t,x) \geq 0$ in $\domain$ and 
\begin{equation} \label{4.7}
 |\partial_{x}^{\alpha} h_{s}(t,x)| \leq C_{\alpha},\ |\alpha| \geq 1,
\end{equation}
\begin{equation} \label{4.8}
|\partial_{x}^{\alpha} w_{s}(t,x)| \leq C_{\alpha},\ |\alpha| \geq 1
\end{equation}
in $\domain$, where $|\Omega|$ denotes the Hilbert-Schmidt norm $\bigl( \sum_{i,j=1}^l|\Omega_{ij}|^2\bigr)^{1/2}$ of  a matrix $\Omega =(\Omega_{ij};i\downarrow j \rightarrow 1,2,\dots,l)$. Then we have 
\begin{equation} \label{4.9}
|\partial_{x}^{\alpha}\partial_{y}^{\beta} \mathcal{F}(t',s';\qts)| \leq C_{\alpha\beta}(t' - s'),\ |\alpha+ \beta| \geq 1
\end{equation}
for $0 \leq s \leq s' \leq t' \leq t \leq T$ and $(x,y) \in \bR^{2d}$.
\end{lem}
\begin{proof}
(1) We set $\mathcal{U}(t) = \mathcal{F}(t,s;q)$.  From \eqref{2.13} we have
\begin{align} \label{4.10}
& \frac{d}{d\theta}\, \mathcal{U}(\theta)^{\dag}\,\mathcal{U}(\theta) = \mathcal{U}(\theta)^{\dag}\bigl\{iH_s(\theta,q(\theta)) - W_s(\theta,q(\theta))\bigr\}\mathcal{U}(\theta)
\notag \\ 
&\quad - \mathcal{U}(\theta)^{\dag}\bigl\{iH_s(\theta,q(\theta)) + W_s(\theta,q(\theta))\bigr\}\mathcal{U}(\theta) 
\notag \\
& = -2 \mathcal{U}(\theta)^{\dag}W_s(\theta,q(\theta))\mathcal{U}(\theta) \leq 0.
\end{align}
Hence we have \eqref{4.5} because of $\mathcal{U}(s) = I$.  Taking $e_1 = {}^t(1,0,\dots,0) \in \bR^l$, from \eqref{4.5} we have
\begin{equation*} 
1 \geq \, <\mathcal{U}(t)^{\dag}\,\mathcal{U}(t)e_1,e_1>\, = \sum_{i=1}^l |\mathcal{F}_{i1}(t,s;q)|^2,
\end{equation*}
where $<\cdot,\cdot>$ denotes the usual inner product of $\bR^l$.  In the same way we can prove \eqref{4.6}. 
\par
(2) From \eqref{2.13} we can easily see 
\begin{align} \label{4.11}
   & \frac{\partial}{\partial x_j}\mathcal{F} (t',s';\qts)  = -\int_{s'}^{t'} \mathcal{F} (t',\theta;\qts)
   \left[\frac{\partial}{\partial x_j}\bigl\{
iH_s(\theta,\qts(\theta)) +  W_s(\theta,\qts(\theta)) \bigr\}\right]
\notag \\
& \qquad 
\times \mathcal{F} (\theta,s';\qts)d\theta
\end{align}
(cf. (3.3) in \cite{Ichinose 2007}).  Then, noting \eqref{3.1}, from \eqref{4.6} - \eqref{4.8} we can prove 
\begin{equation*} 
|\partial_{x_j} \mathcal{F}(t',s';\qts)| \leq C(t' - s')
\end{equation*}
with a constant $C \geq 0$. In the same way we can prove \eqref{4.9} from \eqref{4.11} by induction.
\end{proof}
\begin{lem}
Assume $W_s(t,x) \geq 0$, \eqref{1.11} and \eqref{2.16}.  Then we have
\begin{equation} \label{4.12}
\left|\partial_x^{\alpha}\partial_y^{\beta}\left\{ \calF 
(t,s;q^{t,s}_{x,y})- I \right\}
\right| \leq C_{\alpha,\beta} (t - s)
\end{equation}
in $0 \leq s \leq t \leq T$ and $(x,y) \in \bR^{2d}$ for all $\alpha$ and $\beta$.
\end{lem}
\begin{proof}
From \eqref{2.13} we have
\begin{equation*} 
\calF (t',s;q^{t,s}_{x,y})- I =
- \int_{s}^{t'}\bigl\{ iH_s(\theta,\qts (\theta)) + W_s(\theta,\qts (\theta))\bigr\}\calF 
(\theta,s;q^{t,s}_{x,y})d\theta.
\end{equation*}
Hence, by \eqref{1.11} , \eqref{2.16} and \eqref{4.6} we see
\begin{equation*} 
\left|\calF 
(t,s;q^{t,s}_{x,y})- I 
\right| \leq C (t - s)
\end{equation*}
in $0 \leq s \leq t \leq T$ and $(x,y) \in \bR^{2d}$ with a constant $C \geq 0$.  The inequalities \eqref{4.12} for $|\alpha + \beta| \geq 1$ follow from \eqref{4.9}.
\end{proof}
\begin{pro}
Besides the assumptions of Proposition 3.5 we suppose $W_s(t,x) \geq 0$,  \eqref{1.11} and \eqref{2.16}.  Let $\rho^* > 0$ be the constant determined in Lemma 3.3 and $\mathcal{C}_s(t,s)$ the operator defined by \eqref{4.2}.  Then there exists a constant $K'_0 \geq 0$ such that
\begin{equation}   \label{4.13}
 \Vert \Cts_s(\ts)f\Vert \leq e^{K'_0(t-s)}\Vert f \Vert, \quad 0 \leq t - s \leq \rho^*
\end{equation}
           for all $f \in (L^2)^l$.
\end{pro}
\begin{proof}
Using $\Cts(t,s)$ defined by \eqref{3.7}, we write
\begin{align} \label{4.14}
& \Cts_s(t,s) f = \Cts(t,s)f +
\sqrt{\frac{m}{2\pi i\rho}}^{\, d} \int \bigl(\exp iS_w(t,s; 
q^{t,s}_{x,y})\bigr)
  \notag \\
& \quad \times  
\left\{\calF (t,s;\qts) - I\right\}f(y)dy
\end{align}
for $f \in \Sspace^l$.   Noting \eqref{4.12}, from Proposition 3.6 we see that the $L^2$-norm of the second term on the right-hand side of \eqref{4.14} is bounded by $C(t-s)\Vert f \Vert$ from above with a constant $C \geq 0$. Proposition 3.5 is showing \eqref{3.23}. Hence we have
\begin{align*}
\Vert \Cts_s(t,s) f \Vert &
\leq \mathrm{e}^{K_0(t - s)}\Vert f \Vert + C(t - s)\Vert 
f \Vert \\
& \leq \mathrm{e}^{(K_0+C)(t - s)}\Vert f \Vert, \quad 0 \leq t - s \leq \rho^*.
\end{align*}
Consequently, we can prove \eqref{4.13} with a constant $K'_0 \geq 0$.
\end{proof}
\begin{lem}
Besides the assumptions of Lemma 4.1 we assume $W_s(t,x) \geq 0$ and \eqref{4.7} - \eqref{4.8}.  Then, there exist functions $r_{ij}(t,s;x,w)\ (i,j = 1,2,\dots,l)$ satisfying \eqref{3.5} for an integer $M \geq 0$ such that $\partial_w^{\alpha}\partial_x^{\beta}r_{ij}(t,s;x,w)$ are continuous in $0 \leq s \leq t \leq T$ and $(x,w) \in \bR^{2d}$ for all $\alpha$ and $\beta$, and we have
\begin{align}  \label{4.15}
     & \left\{i\frac{\partial}{\partial t} - H_w(t) -H_s(t,x)  + iW_s(t,x)\right\}\Cts_s(\ts) f
     \notag \\
& = \sqrt{t - s} \Bigl(R_{ij}(\ts);i\downarrow j\rightarrow 1,2,\dots,l \Bigr)f \equiv \sqrt{t - s} R(t,s)f
\end{align}
for $f \in \Czerospace^l$, where $R_{ij}(\ts)$ are the operators defined by \eqref{3.6}.
\end{lem}
\begin{proof}
We note that  \eqref{4.6} and \eqref{4.9} hold under our assumptions.  Consequently, replacing $V(t,x)$ and $H_s(t,x)$ with $V(t,x) - iW(t,s)$ and $H_s(t,x) - iW_s(t,x)$, respectively in the proof of Proposition 3.5 of \cite{Ichinose 2007}, we can complete the proof of Lemma 6.4. In particular, see (3.21)  and (3.22) of \cite{Ichinose 2007}.
\end{proof}
\begin{pro}
Besides the assumptions of Proposition 4.2, we suppose $W_s(t,s) \geq 0$ and \eqref{4.7} - \eqref{4.8}.  Then, there exist functions $r_{ij}(t,s;x,w)\ (i,j = 1,2,\dots,l)$ satisfying the properties stated in Lemma 6.4 and we have
\begin{equation} \label{4.16}
\Vert R(t,s)f\Vert_a \leq C_a\Vert f\Vert_{M+aM_1}, \quad 0 \leq t - s \leq \rho^*
\end{equation}
for $a = 0, 1,2,\dots$ and all $f \in (B^{M+aM_1})^l$, where $M_1$ is the integer in \eqref{3.22}.
\end{pro}
\begin{proof}
As in the proof of Proposition 4.2, we can easily see that the assumptions of Lemma 6.4 hold.  Hence, using Lemma 6.4, from Proposition 3.6 we can prove \eqref{4.16}.
\end{proof}
\begin{pro}
Besides the assumptions of Theorem 2.2 we suppose \eqref{1.11}, \eqref{2.16} and \eqref{4.1}.  Then, for $a = 0, 1, 2, \dots$ there exist constants $K'_a \geq 0$ such that 
\begin{equation} \label{4.17}
\Vert C_s(t,s)f\Vert_a \leq e^{K'_a(t - s)}\Vert f\Vert_{a}, \quad 0 \leq t - s \leq \rho^*
\end{equation}
for all $f \in (B^{a})^l$.
\end{pro}
\begin{proof}
We note that \eqref{4.12} hold.  Then, applying Proposition 3.6 as $M_1 = 1$ to the second term on the right-hand side of \eqref{4.14}, its $B^a$-norm is bounded by $C_a(t - s)\Vert f \Vert_a$ from above  with a constant $C_a \geq 0$ for $0 \leq t - s \leq \rho^*$.  Hence, using \eqref{3.37}, from \eqref{4.14} we can prove 
\begin{align*} 
\Vert C_s(t,s)f\Vert_a &\leq e^{K_a(t - s)}\Vert f\Vert_{a} + C_a(t - s)\Vert f\Vert_{a}  
\\ & \leq e^{(K_a + C_a)(t - s)}\Vert f\Vert_{a},
\quad 0 \leq t - s \leq \rho^*,
\end{align*}
which shows \eqref{4.17}.
\end{proof}
\par
\vspace{0.5cm}
{\bf Proofs of Theorems 2.3 and 2.4.}  Since we are assuming \eqref{1.11} and \eqref{2.16}, we obtain the same results as in Theorem 4.5 for the equation \eqref{1.10} from (1) of Theorem 2.1 of \cite{Ichinose 2019}. We write the solution to \eqref{1.10} with $u(s) = f$ as $U_s(t,s)f$. Then, using Proposition 6.5, we can prove
 \begin{equation} \label{6.18}
 \Vert \Cts_s(\ts)f - U_s(t,s)f\Vert_a \leq  C_a\rho^{3/2}\Vert f\Vert_{M+a},\quad  0\leq  t-s \leq \rho^*
 \end{equation}
 for $a = 0,1,2,\dots$ as in the proof of \eqref{3.41}.  Thereby, following the proofs of Theorems 2.1 and 2.2 in Sect. 5, we can complete the proofs of Theorems 2.3 and 2.4 together with \eqref{4.13} and \eqref{4.17}.
\par
\vspace{0.5cm}
{\bf Proofs of Theorems 2.5 and 2.6.}
We may assume $W_{j}(t,\bx_{j}) \geq  0$ and $W_{sj}(t,\bx_j) \geq 0\ (j = 1,2,\dots,N)$ without loss of generality as in the proofs of Theorems 2.3 and 2.4.
For a continuous path $q(\theta) = \bigl(\bq_1(\theta),\dots,
\bq_N(\theta) \bigr) \in \bR^{dN}\ (s \leq \theta \leq t)$ we define $\mathcal{F}^{\sharp}(\theta,s;q)\ (s \leq \theta \leq t)$
by the solution to 
\begin{align} \label{4.18}
&\frac{d}{d\theta}\,\mathcal{U}^{\sharp}(\theta) = 
- \biggl[\sum_{j=1}^N  I_1\otimes \cdots \otimes I_{j-1}\otimes \bigl\{iH_{sj}(\theta,\bq_j(\theta)) +  W_{sj}(\theta,\bq_j(\theta))\bigl\} \notag\\
& \quad   \otimes I_{j+1}\otimes \cdots\otimes I_N 
\biggr]\mathcal{U}^{\sharp}(\theta), \quad \mathcal{U}^{\sharp}(s) = I_1\otimes \cdots \otimes I_N
\end{align}
in the same way as we do $\mathcal{F}(\theta,s;q)$  from \eqref{2.13}.
\par
	We consider $\mathcal{F}_j(\theta,s;\bq_{j})$ in \eqref{2.19}.  Then from the simple properties of the tensor products (cf. 4.2.1 and 4.2.10 in \S 4.2 of \cite{Horn-Johnson}, and \S VIII.10 of \cite{Reed-Simon I}) 
	 we can easily have
\begin{align} \label{6.20}
& \frac{d}{d\theta}\, \mathcal{F}_{1}(\theta,s;\mathbf{q}_1) \otimes \cdots \otimes  \mathcal{F}_{N}(\theta,s;\mathbf{q}_N) 
 = \sum_{j=1}^N \mathcal{F}_{1}(\theta,s;\mathbf{q}_1)\otimes \cdots \otimes \mathcal{F}_{j-1}(\theta,s;\mathbf{q}_{j-1}) \notag \\
&  \otimes \frac{d}{d\theta}\, \mathcal{F}_{j}(\theta,s;\mathbf{q}_j) \otimes \mathcal{F}_{j+1}(\theta,s;\mathbf{q}_{j+1})\otimes \cdots  \mathcal{F}_{N}(\theta,s;\mathbf{q}_N) = -\sum_{j=1}^N \mathcal{F}_{1}(\theta,s;\mathbf{q}_1)\otimes \notag\\
&\cdots \otimes \mathcal{F}_{j-1}(\theta,s;\mathbf{q}_{j-1})\otimes \bigl\{iH_{sj}(\theta,\bq_{j}(\theta)) + W_{sj}(\theta,\bq_{j}(\theta))\bigr\}\mathcal{F}_{j}(\theta,s;\mathbf{q}_j)
\notag\\
& \otimes\mathcal{F}_{j+1}(\theta,s;\mathbf{q}_{j+1})\otimes \cdots  \mathcal{F}_{N}(\theta,s;\mathbf{q}_N) 
=  -\sum_{j=1}^N \Bigl[I_{1}\otimes \cdots\otimes I_{j-1}\otimes \bigl\{iH_{sj}(\theta,\bq_{j}(\theta)) \notag\\
& + W_{sj}(\theta,\bq_{j}(\theta))\bigr\}\otimes I_{j+1}\otimes \cdots\otimes I_{N}\Bigr]\mathcal{F}_{1}(\theta,s;\mathbf{q}_1) \otimes \cdots \otimes  \mathcal{F}_{N}(\theta,s;\mathbf{q}_N).
\end{align}
Consequenlty we have
\begin{equation} \label{4.19}
\mathcal{F}^{\sharp}(\theta,s;q) = \mathcal{F}_1(\theta,s;\bq_{1})\otimes \cdots \otimes \mathcal{F}_N(\theta,s;\bq_{N}),
\end{equation}
which follows from uniqueness of the solutions to \eqref{4.18}.  Hence we can write \eqref{2.19} as 
\begin{equation} \label{6.22}
\exp *iS^{\sharp}_{sw}(t,0;\qdelta) = \bigl(\exp iS^{\sharp}_{w}(t,0;\qdelta)\bigr) \mathcal{F}^{\sharp}(\theta,s;\qdelta),
\end{equation}
which corresponds to \eqref{2.14} for one particle system.
We set 
\begin{align}  \label{6.23}
& H^{\sharp}_s(t,x) := \sum_{j=1}^N I_1\otimes \cdots \otimes I_{j-1} 
  \otimes H_{sj}(t,\bx_j)\otimes I_{j+1}\otimes \cdots\otimes I_N, \\
& W^{\sharp}_s(t,x) := \sum_{j=1}^N I_1\otimes \cdots \otimes I_{j-1} 
  \otimes W_{sj}(t,\bx_j)\otimes I_{j+1}\otimes \cdots\otimes I_N,   \\
  & W^{\sharp}(t,x) =  \sum_{j=1}^NW_{j}(t,\bx_{j}).
\end{align}
Then we can write \eqref{2.18} in the form of \eqref{1.9} as
\begin{equation} \label{6.26}
\mathcal{L}^{\sharp}_{sw}(t,x,\dot{x}) = \mathcal{L}^{\sharp}_{w}(t,x,\dot{x}) -  H^{\sharp}_s(t,x) +  iW^{\sharp}_s(t,x).
\end{equation}
We can easily see that both of  $H^{\sharp}_s(t,x)$ and $W^{\sharp}_s(t,x)$ are written as $l^{N}\times l^{N}$ Hermitian matrices (cf. 4.2.5 in \S 4.2 of \cite{Horn-Johnson} and \S VIII.10 of \cite{Reed-Simon I}) and satisfy \eqref{1.11},  \eqref{2.16} and $W^{\sharp}_{s}(t,x) \geq 0$.  Hence we can obtain the same results as in Lemma 6.1 for $\mathcal{F}^{\sharp}(t,s;q)$.
\par
	We consider $C_{w^{\sharp}}(t,s;x,y)$ defined by \eqref{3.9} where $W = W^{\sharp}(t,x)$.  Then we can easily see the same estimates as in (3.19) for $C_{w^{\sharp}}(t,s;x,y)$ because of $C_{w^{\sharp}}(t,s;x,y) = \prod_{j=1}^{N}C_{w_{j}}(t,s;\bx_{j},\mathbf{y}_{j})$ and $W_{j}(t,\bx_{j}) \geq 0\ (j = 1,2,\dots,N)$.  We also note that $W^{\sharp}(t,x)$ satisfies \eqref{2.7} and \eqref{2.9}.  Hence, using the results stated above for $\mathcal{F}^{\sharp}(t,s;q)$ and $C_{w^{\sharp}}(t,s;x,y)$ and following the proofs of Theorems 2.5 and 2.6, we can complete the proofs of Theorems 2.5 and 2.6 from \eqref{6.22}.
	\par
  %
{\bf  Data availability statement.} Data sharing is not applicable to this
article as no new data were created or analyzed in this study.
 %
%
 \hspace{5cm}Department of Mathematics, Shinshu University,\\
\hspace{5cm}Matsumoto 390-8621, Japan. \\
 \hspace{5cm}E-mail: ichinose@math.shinshu-u.ac.jp%

\end{document}